\documentclass[submission,copyright,creativecommons]{eptcs}
\providecommand{\event}{TARK 2017} %
\usepackage{breakurl}             %

\usepackage{url}
\usepackage{graphicx}
\usepackage{latexsym}
\DeclareMathAlphabet{\mathcal}{OMS}{cmsy}{m}{n}

\usepackage{scalefnt}
\usepackage{tikz}
\usetikzlibrary{trees}
\usetikzlibrary{shadows,patterns,shapes}

\usepackage{enumerate}
\usepackage{graphicx}
\usepackage{color}

\usepackage{amsmath}
\usepackage{amsfonts}
\usepackage{amssymb,bm}
\usepackage{verbatim}
\usepackage{microtype}
\newcommand{\bfe}[1]{\begin{bfseries}\emph{#1}\end{bfseries}\index{#1}}

\newcommand{\La}{\mbox{$\:\Leftarrow\:$}}
\newcommand{\Ra}{\mbox{$\:\Rightarrow\:$}}

\newcommand{\sse}{\mbox{$\:\subseteq\:$}}

\newcommand{\fa}{\mbox{$\forall$}}

\newcommand{\LL}{\mbox{$\ldots$}}

\newcommand{\C}[1]{\mbox{$\{{#1}\}$}}           %

\newcommand{\NI}{\noindent}
\newcommand{\HB}{\hfill{$\Box$}}

\newcommand{\III}{\vspace{3 mm}}
\newcommand{\II}{\vspace{2 mm}}

\newcommand{\szkew}[1]{\relax \setbox0=\hbox{\kern -24pt $\displaystyle#1$\kern 0pt }%
\box0}
{\catcode`\@=11 \global\let\ifjusthvtest@=\iffalse}
\newcounter{oldmycaption}

\newcommand{\weg}[1]{}

\newcommand{\Agents}{\mathsf{A}}

\newcommand{\Atoms}{\mathsf{P}}
\newcommand{\Atomsb}{\mathsf{Q}}

\newtheorem{theorem}{Theorem}

\newtheorem{lemma}[theorem]{Lemma}
\newtheorem{note}[theorem]{Note}
\newtheorem{definition}[theorem]{Definition}
\newtheorem{corollary}[theorem]{Corollary}

\newtheorem{exam}[theorem]{Example}
\newenvironment{example}{\begin{exam} \rm}{\end{exam}}

\newcommand{\qed}{\hfill$\Box$}

\newenvironment{proof}{\NI {\em Proof.}}{\qed}

\newcommand{\set}[1]{{\{ #1 \}}}

\newcommand{\Model}{\mathcal{M}}

\newcommand{\Calls}{\mathsf{C}}
\newcommand{\Situation}{\mathsf{s}}

\newcommand{\init}{\mathsf{root}}
\newcommand{\Call}{\mathsf{c}}
\newcommand{\Calld}{\mathsf{d}}

\newcommand{\CSequences}{\bm{\Calls}}

\newcommand{\CSequence}{\bm{\Call}}
\newcommand{\CSequenced}{\bm{\Calld}}

\newcommand{\EV}[1]{{\mathsf{EV}}(#1)}
\newcommand{\EVp}[1]{{\mathsf{EPV}}(#1)}

\newcommand{\gs}{*}

\newcommand{\EVVp}{\widetilde{\mathsf EPV}}
\newcommand{\leqs}{\leq_{\mathsf A}}

\newcommand{\leqtwoev}{\leq_{\mathsf{EPV}}}

\title{Common Knowledge in a Logic of Gossips}
\author{Krzysztof R. Apt
	\institute{Centrum Wiskunde \& Informatica\\ Amsterdam, The Netherlands}
	\institute{University of Warsaw\\ Warsaw, Poland}
    \email{k.r.apt@cwi.nl}
	\and
	Dominik Wojtczak
	\institute{University of Liverpool \\
		Liverpool, UK}
	\email{d.wojtczak@liv.ac.uk}
}
\def\titlerunning{Common Knowledge in a Logic of Gossips}
\def\authorrunning{K.~R.~Apt \& D. Wojtczak}

\begin{document}
\maketitle

\begin{abstract}
  Gossip protocols aim at arriving, by means of point-to-point or
  group communications, at a situation in which all the agents know
  each other secrets.  Recently a number of authors
  studied distributed epistemic gossip protocols.  These protocols
  use as guards formulas from a simple epistemic logic, which makes
their analysis and verification substantially easier.

  We study here common knowledge in the context of such a logic.
  First, we analyze when it can be reduced to iterated knowledge. Then
  we show that the semantics and truth for formulas without nested
  common knowledge operator are decidable. This implies that
  implementability, partial correctness and termination of distributed
  epistemic gossip protocols that use non-nested common knowledge
  operator is decidable, as well.  Given that common knowledge is
  equivalent to an infinite conjunction of nested knowledge, these
  results are non-trivial generalizations of the corresponding
  decidability results for the original epistemic logic, established
  in \cite{AW16}.
\end{abstract}
\date{}

\section{Introduction}

Common knowledge is a fundamental notion in epistemic reasoning. It
has its origins in the book of the philosopher David Lewis,
\cite{Lewis:1969}, and the article of the sociologist Morris
Frie\-dell, \cite{Friedell}. By now this concept was applied in many
other fields, including artificial intelligence, psychology, computer
science, game theory, and logic.  An early work on this subject in
computer science and logic is discussed in \cite{FHMV_RAK}. For more
recent accounts and surveys see e.g., \cite{DEV09} and \cite{VS14}.

Study and use of various logics equipped with the common knowledge
operator is a rich field.  As example of recent publications let us just mention
\cite{BenEijKoo05:ckiul}, where an update logic augmented with common
knowledge is investigated, and \cite{wangetal.tark:2009}, where the
correctness of epistemic protocols that rely on common knowledge is
studied.

The purpose of this article is to investigate common knowledge in the
context of a simple epistemic logic proposed in \cite{AGH16} to
express and analyze distributed epistemic gossip protocols.  Gossip
protocols aim at arriving, by means of point-to-point or group
communications, at a situation in which all the agents know each other
secrets, see, e.g., the early survey \cite{HHL88} or the book coverage
\cite{HKPRU05}.  Distributed epistemic gossip protocols were
introduced in \cite{ADGH14}, and further studied in
\cite{ADGH14a,herzig_how_2017,AGH16,van_ditmarsch_parameters_2016,cooper_simple_2016,DEPRS17,CHMMR16},
where in particular various distributed gossiping protocols, their types,
epistemic aspects and objectives, and their interpretation as planning problems
were analyzed.  Such protocols are strikingly simple in their syntax
based on epistemic logic (though not semantics), which makes it easier
to reason about them.

In \cite{AW16} we showed that the distributed epistemic gossip
protocols introduced in \cite{AGH16} are implementable
and proved that the problems of partial correctness and termination of
such protocols are decidable, as well.  In \cite{AKW17} we built upon
these results and showed that the implementability of a distributed
epistemic gossip protocol is a
$\text{P}^{\text{NP}}_{\parallel}$-complete problem, while the
problems of its partial correctness and termination are in
$\text{coNP}^{\text{NP}}$. We also established in \cite{AW17} that fair
termination of the distributed epistemic gossip protocols is
decidable, as well.

In this paper we extend the results of
\cite{AW16} to the language that includes the common knowledge
operator. Given that common knowledge is equivalent to an infinite
conjunction of nested knowledge, these results are non-trivial
generalizations of the previous results.

The obtained results clarify when and how common knowledge can arise
in the context of gossiping.  We prove that three or more agents can
have common knowledge only of true statements. This is not the case
for two agents, even if they do not communicate.  We also show that
under some assumptions common knowledge of two agents coincides with
the 4th fold iterated knowledge.  The main open problem is whether in
this context common knowledge can always be reduced to iterated
knowledge.

\section{Syntax}
\label{sec:syntax}

The purpose of this paper is to analyze common knowledge in the
context of gossip protocols.  To describe it we use a simple modal
language introduced in \cite{AGH16}, though we allow now the common
knowledge operator instead of the agent related knowledge operator.

Throughout the paper we assume a fixed finite set $\mathsf{A}$ of at
least three \bfe{agents}. We assume that each agent holds exactly one
\bfe{secret} and that there exists a bijection between the set of
agents and the set of secrets.  We denote by $\mathsf{P}$ the set of
all secrets.  

Assume a fixed ordering on the agents.  Each \bfe{call} concerns two
different agents, say $a$ and $b$, and is written as $ab$ or $(a,b)$, where agent
$a$ precedes agent $b$ in the assumed ordering.

Calls are denoted by $\Call$, $\Calld$.  Abusing notation we write
$a \in \Call$ to denote that agent $a$ is one of the two agents
involved in the call $\Call$ (e.g., for $\Call := ab$ we have
$a \in \Call$ and $b \in \Call$).

We consider formulas in an epistemic language ${\mathcal L}^{ck}$
defined by the following grammar:
\[
\phi ::= F_a p \mid \neg \phi \mid \phi \land \phi \mid C_G \phi,
\]
where $p \in \mathsf{P}$ and $a \in \mathsf{A}$ and
$G \sse \mathsf{A}$.  Each secret is viewed a distinct constant.  We
denote the secret of agent $a$ by $A$, the secret of agent $b$ by $B$,
where $a,b \in \mathsf{A}$ and $A,B \in \mathsf{P}$, and so on.  When
$G$ is a singleton, say $G = \C{a}$, then we write $C_G$ as $K_a$,
which is the knowledge operator used and studied in the context of
this logic in \cite{AGH16} and \cite{AW16}.

We read $F_a p$ as `agent $a$ is familiar with the secret $p$',
$K_a \phi$ as `agent $a$ knows that formula $\phi$ is true', and
$C_G \phi$ as `the group of agents $G$ commonly knows that formula
$\phi$ is true'.

So $F_a p$ is an atomic formula, while $K_a \phi$ and
$C_G \phi$ are compound
formulas.  
In what follows we shall distinguish the following sublanguages of ${\mathcal L}^{ck}$:

\begin{itemize}
\item ${\mathcal L}_{pr}$, its propositional part, which consists of the formulas
that do not use the $C_G$ modalities,

\item ${\mathcal L}_{wn}$,  which consists of the formulas
without the nested use of the $C_G$ modalities.
\end{itemize}

\section{Semantics}
\label{sec:semantics}

We now recall from \cite{AGH16} semantics of the epistemic
formulas. To this end we recall first the concept of a gossip
situation.

A \bfe{gossip situation} (in short a \bfe{situation}) is a sequence
$\Situation = (\Atomsb_a)_{a \in \Agents}$, where
$\Atomsb_a \sse \Atoms$ for each agent $a$.  Intuitively, $\Atomsb_a$
is the set of secrets $a$ is familiar with in situation $\Situation$.
The \bfe{initial gossip situation} is the one in which each
$\Atomsb_a$ equals ${\{A\}}$ and is denoted by $\init$.
We say that an agent $a$
is an \bfe{expert} in a situation $\Situation$ if he is familiar in
$\Situation$ with all the secrets.
Below sets of secrets will be written down as lists. E.g., the set
$\set{A, B, C}$ will be written as $ABC$. Gossip situations will be
written down as lists of lists of secrets separated by dots. E.g., if
there are three agents, then %
the gossip situation
$(\set{A,B}, \set{A,B}, \set{C}$) will be written as $AB.AB.C$.

Each call transforms the current gossip situation by modifying the set
of secrets the agents involved in the call are familiar with. 
Consider a gossip situation $\Situation := (\Atomsb_d)_{d \in \Agents}$.
Then 
$ab(\Situation) := (Q'_d)_{d \in \Agents}$,
where 
$\Atomsb'_a = \Atomsb'_b = \Atomsb_a \cup \Atomsb_b$,
$\Atomsb'_c = \Atomsb_c$, for $c \neq a,b$.
This simply says that the only effect of a call is that the secrets are shared between
the two agents involved in it.

In \cite{AGH16} computations of the gossip protocols were studied, so
both finite and infinite call sequences were used. Here we focus on
the finite call sequences as we are only interested in the semantics
of epistemic formulas.  So to be brief, unless explicitly stated, a
\bfe{call sequence} is assumed to be finite.

The empty sequence is denoted by $\epsilon$.  We use $\CSequence$ to
denote a call sequence and $\CSequences$ to denote the set of all
finite call sequences.  Given call sequences $\CSequence$ and
$\CSequenced$ and a call $\Call$ we denote by $\CSequence.\Call$ the
outcome of adding $\Call$ at the end of the sequence $\CSequence$ and
by $\CSequence.\CSequenced$ the outcome of appending the sequences
$\CSequence$ and $\CSequenced$.  We write
$\CSequence \sqsubseteq \CSequenced$ to denote the fact that $\CSequenced$
extends $\CSequence$, i.e., that for some $\CSequence'$ we have
$\CSequence.\CSequence' = \CSequenced$.

The result of applying a call sequence to a situation $\Situation$ is
defined inductively as follows:

$\epsilon(\Situation) := \Situation$,

$(\Call.\CSequence)(\Situation) := \CSequence(\Call(\Situation))$.

A gossip situation is a set of possible combinations of secret
distributions among the agents.
As calls progress in sequence from the initial situation, agents may
be uncertain about which one of such secrets distributions is the
actual one. This uncertainty is captured by appropriate equivalence
relations on the call sequences.

\begin{definition} \label{def:model}
A \bfe{gossip model} is a tuple $\Model := (\CSequences,
  \set{\sim_a}_{a \in \Agents})$, where each $\sim_a \subseteq
\CSequences \times \CSequences$  is
the smallest relation such that $\epsilon \sim_a \epsilon$ and
the following conditions hold.
Suppose $\CSequence \sim_a \CSequenced$.

  \begin{enumerate}[(i)]
  \item If $a \not\in \Call$, then $\CSequence.\Call \sim_a \CSequenced$
and $\CSequence \sim_a \CSequenced.\Call$.

\item If $a \in \Call$ and $\CSequence.c(\init)_a = \CSequenced.c(\init)_a$,
then $\CSequence.\Call \sim_a \CSequenced.\Call$.

  \end{enumerate}

A gossip model with a designated call sequence is called a
\bfe{pointed gossip model}.
\end{definition}

So for each set of agents there is exactly one gossip model.
To illustrate the definition of $\sim_a$ note 
for instance that by \emph{(i)} we have $ab, bc \sim_a ab, bd$.
But we do not have $bc, ab \sim_a bd, ab$ since
$(bc, ab)(\init)_a = ABC \neq ABD = (bd, ab)(\init)_a$.

To define semantics of the $C_G$ operator
we use the relation
$\sim_G \subseteq \CSequences \times \CSequences$ defined by
\[
\sim_G = (\cup_{a \in G} \sim_a)^{*},
\]
where $^*$ stands for the transitive reflexive closure of a binary relation.
As stated in \cite{AGH16}, each $\sim_a$ is an equivalence relation. As a result
each $\sim_G$ is an equivalence relation, as well.

Finally, we recall the definition of truth.

\begin{definition} \label{def:semantics}%
Let $(\Model, \CSequence)$ be a pointed gossip model with $\Model :=
(\CSequences, (\sim_a)_{a \in \Agents})$ and $\CSequence \in
\CSequences$. We define the satisfaction relation $\models$
inductively as follows. For convenience we also include the special case of $K_a$ (i.e.,
$G_{\{a\}}$). The clauses for Boolean connectives are as usual and omitted.
\begin{eqnarray*}
(\Model, \CSequence) \models F_a p & \mbox{iff} & p \in \CSequence(\init)_a, \\
(\Model, \CSequence) \models K_a \phi &  \mbox{iff}  & \,\forall \CSequenced \mbox{     s.t.     } \CSequence \sim_a \CSequenced, ~(\Model, \CSequenced) \models \phi, \\
(\Model, \CSequence) \models C_G \phi &  \mbox{iff}  & \forall \CSequenced \mbox{     s.t.     } \CSequence \sim_G \CSequenced, ~(\Model, \CSequenced) \models \phi. 
\end{eqnarray*}
Further
\begin{eqnarray*}
\Model \models \phi &  \mbox{iff}  & \,\forall \CSequence \ (\Model, \CSequence) \models \phi. 
\end{eqnarray*}
When $\Model \models \phi$ we say that $\phi$ \bfe{is true}.
\HB
\end{definition}
So a formula $F_a p$ is true whenever secret $p$ belongs to the set of
secrets agent $a$ is familiar with in the situation generated by the
designated call sequence $\CSequence$ applied to the initial situation
$\init$. The knowledge operator $K_a$ is interpreted as customary in
epistemic logic, using the equivalence relations $\sim_a$, and the
$C_G$ operator is defined as in \cite{FHMV_RAK}.

While ${\mathcal L}^{ck}$ is a pretty standard epistemic language, its
semantics is not.  Indeed, it describes the truth of formulas after a
sequence of calls took place, by analyzing the statements of the form
$(\Model, \CSequence) \models \phi$. So we actually study here a
limited version of a dynamic epistemic logic. To put it differently,
we actually consider statements of the form $[\CSequence] \phi$, where
$[\LL]$ is the standard dynamic logic operator, see, e.g.,
\cite{hvdetal.del:2007}. This explains why the study of the logic ${\mathcal L}^{ck}$
cannot be reduced to a study of a routine epistemic logic.

\section{An alternative equivalence relation}
\label{sec:alternative}

To reason about the $\sim_a$ and $\sim_G$ relations it is easier to
use an alternative equivalence relation between the call sequences
that was introduced in \cite{AW16}. It is based on a concept of a
\bfe{view} of agent $a$ of a call sequence $\CSequence$, written as
$\CSequence_a$, and defined by induction as follows.  
\II

\NI
[Base] 
\[
\epsilon_a := \init,
\]

\NI
[Step]
\[
(\CSequence.\Call)_a :=
\begin{cases} 
\CSequence_a \stackrel{\Call}{\longrightarrow} \Situation &\mbox{if } a \in \Call \\
\CSequence_a  & \mbox{otherwise, }
\end{cases}
\]
where the gossip situation $\Situation$ is defined by putting for $d \in \Agents$
\[
\Situation_d := 
\begin{cases} 
\CSequence.\Call(\init)_{d} &\mbox{if } d \in \Call \\
\Situation'_d  & \mbox{otherwise, }
\end{cases}
\]
where $\Situation'$ is the last gossip situation in $\CSequence_a$.

Intuitively, a view of agent $a$ of a call sequence $\CSequence$ is the
information he acquires by means of the calls in $\CSequence$ he is
involved in. It consists of a sequence of gossip situations connected
by the calls in which $a$ is involved in.  After each such call, say $ab$, agent $a$
updates the set of gossips he and $b$  are currently familiar
with. 

\begin{example} \label{exa:view}
\rm
Let $\Agents = \set{a,b,c}$ and consider the call sequence $(ac,bc,ac)$. 
It generates
the following successive gossip situations starting from $\init$:

\[
A.B.C \stackrel{ac}{\longrightarrow} AC.B.AC \stackrel{bc}{\longrightarrow} AC.ABC.ABC
\stackrel{ac}{\longrightarrow} ABC.ABC.ABC.
\]

We now compare it with the view of agent $a$ of the sequence $(ac,bc,ac)$, which is
\[
A.B.C \stackrel{ac}{\longrightarrow} AC.B.AC \stackrel{ac}{\longrightarrow} ABC.B.ABC.
\]

Thus, in the final gossip situation of this view, agent $b$ is familiar with neither 
the secret $A$ nor $C$.  
\HB
\end{example}
  
We now introduce for each agent $a$ an equivalence relation $\equiv_a$
between the call sequences, defined as follows:
\[
\CSequence \equiv_a \CSequenced \textrm{ iff } \CSequence_a = \CSequenced_a.
\]

So according to this definition two call sequences are equivalent for
agent $a$ if his views of them are the same.
Below we shall rely on the following result from \cite{AW16}.

\begin{theorem}[Equivalence] \label{thm:equiv}
   For each agent $a$ the relations $\sim_a$ and $\equiv_a$ coincide.
\end{theorem}

\section{Semantic matters}

\subsection{General considerations}

We shall need below an alternative definition of truth of the
$C_G \phi$ formulas.
Given a sequence $a_1, \LL, a_k$ of elements of $G$ we abbreviate
$K_{a_1} \LL K_{a_k} \phi$ to $K_{a_1 \LL a_k} \phi$.  We also denote
the set of finite sequences of elements of $G$ by $G^*$.

\begin{note}[\cite{FHMV_RAK}] \label{not:G}
  For all call sequences $\CSequence$ and formulas $C_G \phi \in {\mathcal L}^{ck}$
\[
\mbox{$(\Model, \CSequence) \models C_G \phi$ iff for all $t \in G^*$ 
$(\Model, \CSequence) \models K_t \phi$.}
\]
\end{note}

In other words, the formula $C_G \phi$ is equivalent to the infinite conjunction
$\bigwedge_{t \in G^*} K_t \phi$.
We shall also need the following generalization
of the corresponding result from \cite{AW16} to the logic here studied.

\begin{theorem}[Monotonicity]\label{thm:monog}
  Suppose that $\phi \in {\mathcal L}^{ck}$ is a formula that does not
  contain the $\neg$ symbol.  Then
\[
\mbox{$\CSequence \sqsubseteq \CSequenced$ and $(\Model, \CSequence) \models \phi$
implies $(\Model, \CSequenced) \models \phi$.}
\]
\end{theorem}

\begin{proof}
  By Note \ref{not:G} and Monotonicity Theorem 4 of \cite{AW16}.
\end{proof}

Let us focus now on the case of $\geq 3$ agents. The following result holds.

\begin{theorem} \label{thm:3agents}
Suppose that $|G| \geq 3$. Then for all call sequences $\CSequence$ and formulas $\phi \in {\mathcal L}^{ck}$
\[
  (\Model, \CSequence) \models C_G \phi \mbox{ iff } \Model \models
  \phi.
\]
\end{theorem}
\begin{proof}
  First we prove that for all $\CSequence$ and $\CSequenced$
  \begin{equation}
    \label{equ:iff1}
\CSequence \sim_G \CSequenced.    
  \end{equation}

By the transitivity of $\sim_G$ it suffices to prove that
$\CSequence \sim_G \epsilon$. We prove it by induction on the
length of $\CSequence$. By definition $\epsilon \sim_G \epsilon$. 
Suppose that for some $\CSequence$ we have
$\CSequence \sim_G \epsilon$ and consider a call $\Call$.  Take
$a \in G$ such that $a \not\in \Call$ (it exists since $|G| \geq
3$). Then $\CSequence.\Call \sim_a \CSequence$, so
$\CSequence.\Call \sim_G \epsilon$.

By (\ref{equ:iff1}) we have $(\Model, \CSequence) \models C_G \phi$
iff $\fa \CSequenced \: (\Model, \CSequenced) \models \phi$, which
concludes the proof.
\end{proof}

Theorem \ref{thm:3agents} states that the formulas commonly known by
the agents in a group of at least three agents are precisely the true
formulas.  An example of such a statement is that each agent is
familiar with his secret, i.e., $\bigwedge_{a \in \mathsf{A}} F_a
A$. In contrast, a statement that an agent is familiar with the secret
of another agent, i.e. $F_a B$, where $a \neq b$, is not always true,
so for all call sequences $\CSequence$ we have
$(\Model, \CSequence) \not \models C_G F_a B$, when $|G| \geq 3$.

\subsection{The case of two agents}

The situation changes when the group consists of two agents.  In what follows we 
abbreviate $C_{\{a,b\}}$ to $C_{ab}$.  

\begin{example} \label{exa:2agents}
\rm

\mbox{}

\NI
$(i)$ Consider the formula
\[
\phi := \neg F_a B \lor \bigvee_{c \in \Agents \setminus \{a, b\}} F_c B.
\]
It states that if $a$ is familiar with the secret of $b$, then also another agent
different from $a$ and $b$ is familiar with this secret.
Note that $(\Model, ab) \models \neg \phi$, i.e., $\phi$ is not always true.
We claim that $(\Model, \epsilon) \models  C_{ab} \phi$.

First note that if $\CSequence \sim_a \CSequenced$ or
$\CSequence \sim_b \CSequenced$ and the call $ab$ does not appear in
$\CSequence$, then it does not appear in $\CSequenced$
either. Consequently, if $\CSequence \sim_{\{a,b\}} \epsilon$, then the call
$ab$ does not appear in $\CSequence$.

Conversely, take a call sequence $\CSequence$ such that the call $ab$
does not appear in it. We prove by induction on the length of
$\CSequence$ that $\CSequence \sim_{\{a,b\}} \epsilon$. By definition
$\epsilon \sim_{\{a,b\}} \epsilon$.  Suppose that for some $\CSequence$ we
have $\CSequence \sim_{\{a,b\}} \epsilon$ and consider a call $\Call$.  Either
$a \not \in \Call$ or $b \not \in \Call$, so either
$\CSequence.\Call \sim_a \CSequence$ or
$\CSequence.\Call \sim_b \CSequence$. Consequently
$\CSequence.\Call \sim_{\{a,b\}} \epsilon$.

We conclude that $\CSequence \sim_{\{a,b\}} \epsilon$ iff the call $ab$ does
not appear in $\CSequence$. But for any such $\CSequence$ we have
$(\Model, \CSequence) \models \phi$. This proves that
$(\Model, \epsilon) \models  C_{ab} \phi$.

This shows that even without any call two agents can commonly know a
formula that is not always true.
\II

\NI
$(ii)$
We also have $(\Model, ab) \models C_{ab} (F_a B \land F_b A)$,
i.e., two agents can commonly know some non-trivial information about
themselves.

Indeed, if the call $ab$ appears in $\CSequence$, then
$(\Model, \CSequence) \models F_a B \land F_b A$ and
that if the call $ab$ appears in $\CSequence$ and
$\CSequence \sim_a \CSequenced$ or $\CSequence \sim_b \CSequenced$,
then it also appears in $\CSequenced$.  
\HB
\end{example}

On the other hand for two agents the following partial analogue of 
Theorem \ref{thm:3agents} holds.

\begin{theorem} \label{thm:2} For all call sequences $\CSequence$ that
  do not contain the call $ab$ and all formulas
  $\phi \in {\mathcal L}^{ck}$ that do not contain the $\neg$ symbol
\[
(\Model, \CSequence) \models C_{ab} \phi \mbox{ iff } \Model \models \phi.
\]
\end{theorem}

\begin{proof}
  Suppose $(\Model, \CSequence) \models C_{ab} \phi$. As noticed in
  Example \ref{exa:2agents}$(i)$ $\CSequence \sim_{\{a,b\}} \epsilon$, so
  $(\Model, \epsilon) \models \phi$.  By the Monotonicity Theorem
  \ref{thm:monog} for all call sequences $\CSequenced$ we have
  $(\Model, \CSequenced) \models \phi$. So $\models \phi$.
  Further, $\models \phi$ implies
  $(\Model, \CSequence) \models C_{ab} \phi$ for arbitrary call sequences
  $\CSequence$ and formulas $C_{ab} \phi$.
\end{proof}

Example \ref{exa:2agents}$(ii)$ shows that the restriction that
the call $ab$ does not appear in $\CSequence$ cannot be dropped and
Example \ref{exa:2agents}$(i)$ shows that the claim does not hold
for formulas that do contain the $\neg$ symbol.

Next, we show that for formulas that do not contain the $\neg$ symbol common knowledge
for the group of two agents coincides with the 4th fold iterated knowledge.

Consider an agent $a$ and a call sequence $\CSequence$.  We say that a
call is \bfe{$a$-irrelevant} in $\CSequence$ if its removal does not
affect the view (in the sense of Section \ref{sec:alternative})
of agent $a$ of the call sequence.  Starting from
$\CSequence$ we repeatedly remove from the current call sequence the
first not yet analyzed call if it is $a$-irrelevant and otherwise we keep
it. We call the outcome of such an iteration the
\bfe{$a$-simplification} of $\CSequence$.

\begin{example} \label{exa:sim}
\rm
Suppose, 
\[
\CSequence = bf. cd. bc. ce. df. ef. bh. \underline{af}. bg. \underline{ag}. \underline{ah},
\] 
where for the visibility we underlined the $a$-calls.
Then the $a$-simpli\-fication of $\CSequence$ results in the deletion of the calls $bf$ and $cd$ and equals
\[
bc. ce. df. ef. bh. \underline{af}. bg. \underline{ag}. \underline{ah}.
\]

The views of agent $a$ of both call sequences are as follows.
\II

$A.B.C.D.E.F.D.H \stackrel{af}{\longrightarrow}$

$ABCDEF.B.C.D.E.ABCDEF.G.H \stackrel{ag}{\longrightarrow}$

$ABCDEFGH.B.C.D.E.ABCDEF.ABCDEFGH.H \stackrel{ah}{\longrightarrow}$

$ABCDEFGH.B.C.D.E.ABCDEF.ABCDEFGH.ABCDEFGH$

\HB
\end{example}

Below we say that two calls are \bfe{linked} if exactly one agent
participates in both of them.  Consider now a call sequence
$\CSequence$ with no $a$-irrelevant calls that does not contain the
call $ab$.  We focus on the $b$-calls in $\CSequence$. By the
assumption about $\CSequence$ for each $b$-call $\Call$ in
$\CSequence$
there is a sequence of calls $\Call_1, \LL, \Call_k$ in $\CSequence$
such that

\begin{itemize}
\item $\Call = \Call_1$, 

\item $b \in \Call_1$,

\item for $i \in \{1, \LL, k-1\}$ $a \not \in \Call_i$, 

\item for $i \in \{1, \LL, k-1\}$ the calls $\Call_i$ and $\Call_{i+1}$ are
linked.

\item $a \in \Call_k$.

\end{itemize}

We say then that $\Call_1$ \bfe{leads to} $\Call_k$.  Further, we call
each $b$-call in $\CSequence$ that leads to the earliest possible
$a$-call in $\CSequence$ \bfe{$b$-essential for $a$} and call each
other $b$-call in $\CSequence$ \bfe{$b$-inessential for $a$}.
Intuitively, agent $a$ learns the secret of $b$ only through the
$b$-essential calls. In contrast, he can learn other secrets both
through the $b$-essential and the $b$-inessential calls.

\begin{example}
\rm  

Consider the call sequence
$$bc. ce. df. ef. bh. \underline{af}. bg. \underline{ag}. \underline{ah}$$
from Example
\ref{exa:sim}. The only $b$-essential call for $a$ is $bc$ as it leads
to $af$.
In turn, in the call sequence
$bh. ce. df. ef. \underline{af}. bg. bc. \underline{ag}$ the
$b$-essential calls for $a$ are $bh$ and $bg$ as both of them lead to
$ag$ and no $b$-call leads to the earlier $a$-call $af$.  
\HB
\end{example}

Consider now the first $b$-call in $\CSequence$, call it $bc$, that is
$b$-inessential for $a$ and suppose that $ad$ is the first $a$-call in
$\CSequence$ to which $bc$ leads to, where $c$ and $d$ may coincide.
If $c = d$, then we delete $bc$ from $\CSequence$ and otherwise we
replace $bc$ by $cd$. Intuitively, we `reroute' the information
collected by agent $b$ in the $b$-inessential calls leading to $ad$
using agent $d$. This does not affect agent's $a$ view of the call
sequence since he does not learn the secret of $b$ through the
$b$-inessential calls.

We repeat this operation, starting with $\CSequence$, for all
$b$-calls that are $b$-inessential for $a$ and denote the resulting
call sequence by $R_{ab}(\CSequence)$.

\begin{example}
 \rm

 Consider the call sequence
 $$bc. ce. df. ef. bh. \underline{af}. bg. \underline{ag}. \underline{ah}$$
 from Example
 \ref{exa:sim}. The $b$-inessential calls are $bh$ and $bg$ and both
 lead to $ag$.  So
\[
R_{ab}(bc. ce. df. ef. bh. \underline{af}. bg. \underline{ag}. \underline{ah}) = bc.
 ce. df. ef. gh. \underline{af}. \underline{ag}. \underline{ah}.
\]
\HB
\end{example}

The following lemma establishes the relevant property of $R_{ab}(\CSequence)$.

\begin{lemma} \label{lem:Ra}
  Consider a call sequence $\CSequence$ with no $a$-irrelevant calls
  that does not contain the call $ab$. Then
  $\CSequence \equiv_a R_{ab}(\CSequence)$.
\end{lemma}

\begin{proof}
  Since $\CSequence$ does not contain the call $ab$, the views of $a$
  of $\CSequence$ and $R_{ab}(\CSequence)$ have the
  same sequences of the $a$-calls.

  By definition agent $a$ does not learn the secret of $b$ through the
  $b$-inessential calls for $a$. So the replacement (or possibly
  deletion) of $b$ in all $b$-inessential calls has no effect on the
  status of this secret in the views of $a$ of both sequences.

  Let $bc$ be the first $b$-inessential for $a$ call in $\CSequence$
  and suppose that $ad$ is the first $a$-call in $\CSequence$ it leads
  to.  Let $\CSequence'$ be the outcome of the first step of producing
  $R_{ab}(\CSequence)$.  So it is obtained from $\CSequence$ by
  replacing $bc$ by $cd$ if $c \neq d$, and by deleting $bc$
  otherwise.  We show that the views of the agent $a$ of the call
  sequences $\CSequence$ and $\CSequence'$ are the same.  First, note
  that no $a$-call earlier than $ad$ can be effected by the change in
  $\CSequence$, because $ad$ is the first $a$-call $bc$ leads to.
  
  Consider the $c \neq d$ case first.  The set of secrets an agent is
  familiar with at the same point in $\CSequence$ and $\CSequence'$
  may differ as a result of $bc$ being replaced by $cd$.  However, we
  argue that it is impossible for the agent $a$ to notice this
  difference.  Let us consider a call $ax$ in $\CSequence$, where
  $x \in \Agents$ and the sets of secrets, $S$ and $S'$, agent $e$ is
  familiar before $ax$ is made in $\CSequence$ and $\CSequence'$,
  respectively.
  \begin{itemize}
  \item If $ax$ takes place before $bc$ then $S$ and $S'$ are the same.
  \item If $ax$ is in-between the calls $bc$ and $ad$
  then again the sets $S$ and $S'$ are the same, because 
  $ad$ is the first $a$-call that $bc$ leads to.
  \item If $ax$ is the call $ad$ then we have the following.
  First, just before this call $a$ is already familiar with all the secrets $b$ 
  is familiar with before the call $bc$ is made in $\CSequence$,
  because $ad$ is the first $a$-call $bc$ leads to.
  So $a$ is still familiar with all these secrets in $\CSequence'$. Thus
 the only secrets that may be lost by replacing $bc$ by $cd$ 
  are the ones that $c$ is familiar with at that point.
  However, these secrets are passed to $d$ and $a$ 
  still learns them all in $\CSequence'$ through the call $ad$.
\item If $ax$ takes place after the call $ad$ then the difference
  between $S$ and $S'$ could be at most in the set of secrets $a$
  learned through the call $bc$.  However, $a$ already knows these
  secrets after the call $ad$ is made, so $S = S'$.\footnote{Note
    that this crucially depends on the fact that from any call each
    involved agent learns the union of the sets of secrets the callers
    are familiar with and not the set of secrets the other caller is
    familiar with.}
  \end{itemize}
The reasoning for the $c = d$ case is completely analogous and omitted.

By iterating the above argument, starting with
$\CSequence$, we obtain $R_{ab}(\CSequence)$ without affecting the
view of the agent $a$.
\end{proof}

The following consequence of the above lemma is crucial.
Here $\equiv_{abab}$ stands for the the composition of the relations
$\equiv_{a}$, $\equiv_{b}$, $\equiv_{a}$ and $\equiv_{b}$, i.e.,
\[
\equiv_{abab} \ = \ \equiv_{a} \circ \equiv_{b} \circ \equiv_{a} \circ \equiv_{b},
\]
where $\circ$ is the composition of two binary relations.

\begin{theorem} \label{thm:abab}
Consider a call sequence $\CSequence$ that does not contain the call $ab$.
Then 
\[
\CSequence \equiv_{abab} \epsilon.
\]
\end{theorem}

\begin{proof}
  Let $\CSequence_1$ be the $a$-simplification of $\CSequence$ and
  $\CSequence_2 = R_{ab}(\CSequence_1)$.  Then
  $\CSequence \equiv_{a} \CSequence_1$ and by Lemma \ref{lem:Ra}
  $\CSequence_1 \equiv_{a} \CSequence_2$, so
  $\CSequence \equiv_{a} \CSequence_2$.

  If there are no $b$-calls in $\CSequence_2$, then
  $\CSequence_2 \equiv_{b} \epsilon$, so
  $\CSequence \equiv_{ab} \epsilon$, from which the conclusion follows
  since $\epsilon \equiv_{ab} \epsilon$.  Otherwise let $\CSequence_3$
  be the prefix of $\CSequence_2$ that ends with the last
  $b$-call. Then $\CSequence_2 \equiv_{b} \CSequence_3$.

  All the $b$-calls in $\CSequence_3$ are $b$-essential for $a$ in the
  call sequence $\CSequence_2$ and the $a$-call through which agent
  $a$ learns in $\CSequence_2$ the secret of $B$ is located after all
  these $b$-essential calls. %
  It follows that 
  $\CSequence_3 \equiv_{a} \CSequence_4$, where $\CSequence_4$ is the
  result of removing all $b$-calls from $\CSequence_3$, 
  because no $b$-call in $\CSequence_3$ can possibly lead to an $a$-call
  in $\CSequence_3$.

  Now, $\CSequence_4$ does not contain any $b$-calls, so
  $\CSequence_4 \equiv_b \epsilon$. This proves the claim.
\end{proof}

\begin{example}
The following example illustrates the call sequences generated in the above proof.
Let
\[
\begin{array}{l}
\CSequence = \underline{ah}. cd. bc. bd. be. \underline{ad}. bf. bg. \underline{af}.
\end{array}
\]
Then the $a$-simplification of $\CSequence$ results in removing the calls $cd$ and $bg$ and
equals
\[
\begin{array}{l}
\CSequence_1 = \underline{ah}. bc. bd. be. \underline{ad}. bf. \underline{af}.
\end{array}
\]
Subsequently $R_{ab}(\CSequence_1)$ equals
\[
\begin{array}{l}
\CSequence_2 = \underline{ah}. bc. bd. ef. \underline{ad}. \underline{af}.
\end{array}
\]
Next, the prefix of $\CSequence_2$ that ends with the last $b$-call is
\[
\begin{array}{l}
\CSequence_3 = \underline{ah}. bc. bd.
\end{array}
\]
Finally, the result of removing all $b$-calls from $\CSequence_3$ equals
\[
\begin{array}{l}
\CSequence_4 = \underline{ah}
\end{array}
\]
and $\CSequence_4 \equiv_b \epsilon$.
\HB
\end{example}

This brings us to the following conclusion.

\begin{theorem} \label{thm:abab1} For all call sequences $\CSequence$
  that do not contain the call $ab$ and all formulas
  $\phi \in {\mathcal L}^{ck}$ that do not contain the $\neg$ symbol
\[
(\Model, \CSequence) \models K_{abab} \phi \mbox{ iff } \Model \models \phi
\]
and
\[
(\Model, \CSequence) \models K_{baba} \phi \mbox{ iff } \Model \models \phi.
\]
\end{theorem}

\begin{proof}
  By symmetry it suffices to prove the first equivalence.  By Theorem
  \ref{thm:abab} $(\Model, \CSequence) \models K_{abab} \phi$ implies
  $(\Model, \epsilon) \models \phi$, so by the Monotonicity Theorem
  \ref{thm:monog} for all call sequences $\CSequenced$ we have
  $(\Model, \CSequenced) \models \phi$, i.e., $\Model \models \phi$. 

  Further, for arbitrary call sequences $\CSequence$ and formulas
  $\phi \in {\mathcal L}^{ck}$, $\Model \models \phi$ implies
  $(\Model, \CSequence) \models K_{abab} \phi$.
\end{proof}

\begin{corollary} \label{cor:ck}
  For all call sequences $\CSequence$
  that do not contain the call $ab$ and all formulas $\phi \in {\mathcal L}^{ck}$
that do
  not contain the $\neg$ symbol
\[
(\Model, \CSequence) \models C_{ab} \phi \mbox{ iff } (\Model, \CSequence) \models K_{abab} \phi.
\]
\end{corollary}
\begin{proof}
  By Theorems \ref{thm:abab1} and \ref{thm:2}.
\end{proof}

This together with Note \ref{not:G} shows that under the
conditions of the above corollary the infinite conjunction
$\bigwedge_{t \in G^*} K_t \phi$ is equivalent to the 4th fold
iterated knowledge $ K_{abab} \phi$.  To conclude this analysis we now
show that common knowledge for two agents is not equivalent to any
shorter knowledge iteration.

\begin{corollary} \label{cor:aba}
Let $\CSequence = ac, bc, ac$. Then
$(\Model, \CSequence) \models K_{aba} F_c A$, 
$(\Model, \CSequence) \models K_{ab} F_c A$ and
$(\Model, \CSequence) \models K_{a} F_c A$.
Further, $(\Model, \CSequence) \not \models C_{ab} F_c A$.
\end{corollary}

\begin{proof}
Take $\CSequenced$ such that $\CSequence \sim_a
\CSequenced$. Then $\CSequenced$ is of the form
$\CSequence_1, ac, \CSequence_2, bc,$ $\CSequence_3, ac, \CSequence_4$. Next take
$\CSequenced_1$ such that $\CSequenced_1 \sim_b \CSequenced$. Then
$\CSequenced$ is of the form $\CSequence_5, ac, \CSequence_6, bc, \CSequence_7$. Next take
$\CSequenced_2$ such that $\CSequenced_2 \sim_a \CSequenced_1$. Then
$\CSequenced_2$ is of the form $\CSequence_8, ac, \CSequence_9$.  So
$(\Model, \CSequenced_2) \models F_c A$, which implies the claim.

The next two claims follow since for all formulas
$\phi \in {\mathcal L}^{ck}$, $K_{aba} \phi$ implies both
$K_{ab} \phi$ and $K_{a} \phi$.

Now note that for all sequences $t \in \{a, b\}^*$ that extend $abab$
or $baba$ and all formulas $\phi \in {\mathcal L}^{ck}$, $K_t \phi$
implies $K_{abab} \phi$ or $K_{baba} \phi$.  So Theorem
\ref{thm:abab1} implies that for all such sequences $s$ and all call
sequences $\CSequence$ that do not contain the call $ab$
\[
(\Model, \CSequence) \not \models K_{s} F_c A.
\]
Indeed, the formula $F_c A$ is not always true.

Further, by Note \ref{not:G} we also have for such call sequences
$\CSequence$
\[
(\Model, \CSequence) \not \models C_{ab} F_c A.
\]
In particular, this holds for the above call sequence
$\CSequence = ac, bc, ac$. 
\end{proof}

\section{Decidability issues} %
\label{sec:ck2}

\subsection{Decidability of semantics}

We begin with the decidability of the semantics.  
First we establish some properties of the semantics of common
knowledge of two agents.

Consider a call $ab$ and a call sequence $\CSequence$.
Starting from $\CSequence$ we repeatedly remove from the current call
sequence a redundant call that differs from $ab$.  We call each
outcome of such an iteration an \bfe{$ab$-reduction} of $\CSequence$.
Further, we say that a call sequence $\CSequence$ is
\bfe{$ab$-redundant free} if no call from $\CSequence$ that differs
from $ab$ is redundant in it. Clearly each $ab$-reduction is
$ab$-redundant free.

\begin{corollary} \label{cor:ab-reduction}
  Let $\CSequenced$ be an $ab$-reduction of $\CSequence$.  
  Then 
  \begin{enumerate}[(i)]
  \item $\CSequence \sim_{\{a,b \}} \CSequenced$,
  \item  for all   formulas $\phi \in {\mathcal L}_{pr}$, $(\Model, \CSequence) \models \phi$ iff $(\Model, \CSequenced) \models \phi$. 
  \end{enumerate}
\end{corollary}

\begin{lemma} \label{lem:ab-finite}
  For each call $ab$ and a  call sequence $\CSequence$ the set of $ab$-redundant free
  call sequences $\CSequenced$ such that $\CSequence \sim_{\{a,b\}} \CSequenced$ is finite.
\end{lemma}

\begin{proof}
  Consider an $ab$-redundant free call sequence $\CSequenced$ such that
  $\CSequence \sim_{\{a, b\}} \CSequenced$. Then $\CSequenced$ has the same
  number, say $k$, of calls $ab$ as $\CSequence$.
  
  Associate with $\CSequenced$ the sequence of gossip situations
  $$\CSequenced^0(\init), \CSequenced^1(\init), \ldots,
  \CSequenced^m(\init),$$ 
where $m$ is the length of $\CSequenced$,
  $\CSequenced^0 = \epsilon$, and
  $\CSequenced^k = \Calld_1, \Calld_2, \ldots, \Calld_k$ for
  $k = 1, \dots, m$. This sequence monotonically grows, where we
  interpret the inclusion relation componentwise. Moreover, for all
  calls $\Calld_i$ different from $ab$ the corresponding
  inclusion is strict.  Consequently, $m$, the length of
  $\CSequenced$, is bounded by $k + |\Agents|^{|\Agents|}$, the sum of
  the number of calls $ab$ in $\CSequence$ and of the total number of
  secrets in the gossip situation in which each agent is an expert.

But for each $m$ there are only finitely many call sequences
of length at most $m$.  This concludes the proof.
\end{proof}

We can now prove the desired result.

\begin{theorem}[Decidability of Semantics] \label{thm:ck-decidability}
  For each call sequence $\CSequence$ it is decidable 
  whether for a formula $\phi \in {\mathcal L}_{wn}$, $(\Model,
  \CSequence) \models \phi$ holds.
\end{theorem}

\begin{proof}
  We use the definition of semantics as the algorithm. We only need to
  consider the case of the formulas of the form $C_G \phi$, where
  $\phi \in {\mathcal L}_{pr}$.

If $|G| = 1$, then this is the contents of the Decidability of Semantics 
Theorem 5 of \cite{AW16}.

If $|G| = 2$, say $A = \{a,b\}$, then 
according to Corollary \ref{cor:ab-reduction}
we can rewrite the semantics of $C_G \phi$ as follows:
$$(\Model, \CSequence) \models C_G \phi  \ \ \mbox{  iff  } \ \
\forall \CSequenced \mbox{     s.t.     } \CSequence \sim_G \CSequenced \mbox{ and } \CSequenced \mbox{ is $ab$-redundant free} 
, ~(\Model, \CSequenced) \models \phi,$$ and according to Lemma
\ref{lem:ab-finite} this definition refers to a finite set of call
sequences $\CSequenced$.

If $|G| \geq 3$, then the decidability follows from Theorem \ref{thm:3agents} and
the Decidability of Truth Theorem 6 established in \cite{AW16}.
\end{proof}

This result implies that the gossip protocols that use guards with
non-nested common knowledge operator are implementable.

\subsection{Decidability of truth}
\label{subsec:ver}

Next, we show that truth definition for the formulas of the language
${\mathcal L_{wn}}$ is decidable. Since partial correctness of
gossip protocols with common knowledge operator can be expressed as a
formula of ${\mathcal L_{wn}}$, this implies that the problem of
determining partial correctness of such protocols is decidable.

The key notion is that of an
\bfe{epistemic pair-view}. It is a function of a call sequence $\CSequence$,
denoted by $\EVp{\CSequence}$, defined by 
\begin{itemize}
\item putting for any pair of agents $a,b$:

  $\EVp{\CSequence}(a,b) = \{\CSequenced(\init)\ |\ \CSequence \sim_{\{a,b\}}
  \CSequenced \}$, and setting
\item   $\EVp{\CSequence}(\gs) = \CSequence(\init)$.  
\end{itemize}
So
$\EVp{\CSequence}(a,b)$ is the set of all gossip situations obtained
by means of call sequences that are
$\sim_{\{a,b\}}$--equivalent to
$\CSequence$. 
Further, as $\sim_{\{a,a\}} = \sim_{a}$,
$\EVp{\CSequence}(a,a)$ is the set of all gossip situations consistent
with agent $a$'s observations made throughout $\CSequence$.
Finally, 
$\EVp{\CSequence}(\gs)$ is the actual gossip situation after
$\CSequence$ takes place.  Note that for any $a,b \in \Agents$,
if $\CSequence \sim_{\{a,b\}} \CSequenced$ then
$\EVp{\CSequence}(b,a) = \EVp{\CSequence}(a,b) = \EVp{\CSequenced}(a,b) = \EVp{\CSequenced}(b,a)$.

The following 
holds.

\begin{lemma} \label{lem:ab-eff}
For each call sequence $\CSequence$ and agents $a,b$, the
set $\EVp{\CSequence}(a,b)$ is finite and can be effectively constructed.
\end{lemma}
\begin{proof}
  For any $a \in \Agents$,
  $\EVp{\CSequence}(a,a)$ coincides with the epistemic view 
  $\EV{\CSequence}(a)$, 
  as defined in \cite{AW16}. Hence, we can compute
  $\EVp{\CSequence}(a,a)$ as in Lemma 3 of \cite{AW16}.

Consider now a pair of agents $a,b$ such that $a\neq b$.
  To construct  the set $\EVp{\CSequence}(a,b)$ 
it suffices by Corollary \ref{cor:ab-reduction} 
to consider the $ab$-redundant
  free call sequences $\CSequenced$ and by Lemma \ref{lem:ab-finite}
  there are only finitely many such call sequences $\CSequenced$ for
  which $\CSequenced \sim_{\{a,b\}} \CSequence$.
\end{proof}

Our interest in epistemic pair-views stems
from the following important observation.

\begin{lemma}\label{lem:ck-ev}
  Suppose that $\EVp{\CSequence} = \EVp{\CSequenced}$. 
  Then for all formulas $\phi \in {\mathcal L}_{wn}$, 
  $(\Model, \CSequence) \models \phi$ iff $(\Model, \CSequenced) \models \phi$.
\end{lemma}

\begin{proof}
A straightforward proof by induction shows that for a formula $\psi \in {\mathcal L}_{pr}$ and
arbitrary call sequences $\CSequence'$ and $\CSequenced'$,
\begin{equation}
\mbox{$\CSequence'(\init) = \CSequenced'(\init)$ implies that
  $(\Model, \CSequence') \models \psi$ iff $(\Model, \CSequenced') \models \psi$.}
  \label{equ:2}  
\end{equation}
Since $\EVp{\CSequence}(\gs) = \CSequence(\init)$ and 
$\EVp{\CSequenced}(\gs) = \CSequenced(\init)$, this settles the case for
$\phi = F_a p$.

Next, consider the case of the formulas of the form $C_G \phi$, where
$\phi \in \mathcal L_{wn}$.  

If $|G| = 1$, then $G = \{a\}$ for some $a \in \Agents$ and $C_G$ is
the same as $K_a$.  Since $\EVp{\CSequence} = \EVp{\CSequenced}$
implies $\EV{\CSequence} = \EV{\CSequenced}$, where the epistemic view $\EV{}$ is defined as in \cite{AW16},
the claim follows by Lemma 4 of \cite{AW16}.

If $|G| = 2$, then $G = \{a,b\}$ for some $a,b \in \Agents$.
By (\ref{equ:2}) and the definition of $\EVp{\CSequence}$
$$(\Model, \CSequence) \models C_G \phi  \ \ \mbox{ iff }\ \ 
\forall \CSequence' \mbox{  s.t.  } \CSequence'(\init) \in \EVp{\CSequence}(a,b)
, ~(\Model, \CSequence') \models \phi.$$
So the claim follows since $\EVp{\CSequence}(a,b) = \EVp{\CSequenced}(a,b)$.

If $|G| \geq 3$, then by Theorem \ref{thm:3agents} 
both 
$(\Model, \CSequence') \models C_G \phi$ and $(\Model, \CSequenced') \models C_G \phi$
are equivalent to $\models \phi$.

This settles the case for $C_G \phi$.  The remaining cases of
negation and conjunction follow directly by the induction.
\end{proof}

The above lemma is useful because the epistemic pair-view of each call
sequence is finite, in contrast to the set of call sequences.  Next,
we provide an inductive definition of $\EVp{\CSequence.\Call}(a,b)$
the importance of which will become clear in a moment.

\begin{lemma}\label{lem:ab-composition}
For any call sequence $\CSequence$ and call $\Call = ab$ for agents $a,b \in \Agents$ 
\[
\EVp{\CSequence.\Call}(a,b) = \{\Call(\Situation)\ |\ \Situation \in
\EVp{\CSequence}(a,b), \Call(\Situation)_a =
\Call(\CSequence(\init))_a %
\mbox{ and } \Call(\Situation)_b = \Call(\CSequence(\init))_b 
\}.
\]
\end{lemma}

\begin{proof}
\II

\NI
$(\sse)$ Take $\Situation' \in \EVp{\CSequence.\Call}(a,b)$. By the
  definition of $\EVp{\CSequence.\Call}(a,b)$ there exists a call
  sequence $\CSequenced$ such that
  $\CSequenced.\Call \sim_{\{a,b\}} \CSequence.\Call$ and
  $\Situation' = \CSequenced.\Call(\init)$. 
  So
  $\Situation' = \Call(\Situation)$, where
  $\Situation = \CSequenced(\init)$.
  We prove now that
  $\CSequenced \sim_{\{a,b\}} \CSequence$
  and as a result $\Situation = \CSequenced(\init) \in \EVp{\CSequence}(a,b)$. 

  Note that $\CSequenced.\Call \sim_{\{a,b\}} \CSequence.\Call$
  implies that there exists a 
  sequence $t_1 \ldots t_k \in \{a,b\}^*$ such that
  for some call sequences 
  $\CSequenced_1, \CSequenced_2, \ldots, \CSequenced_{k-1}$
we have
  $\CSequenced.\Call \sim_{t_1} \CSequenced_1.\Call \sim_{t_2} \CSequenced_2.\Call \sim_{t_3} \ldots
  \sim_{t_{k-1}} \CSequenced_{k-1}.\Call \sim_{t_k} \CSequence.\Call$. 
  Note that for any $t \in \{a,b\}$ and call sequences
  $\CSequence', \CSequenced'$ we have that
  $\CSequence'.\Call \sim_{t} \CSequenced'.\Call$ implies
  $\CSequence' \sim_{t} \CSequenced'$, because $\sim_{t}$ is the
  minimal relation satisfying the conditions stated in Definition
  \ref{def:model}.  It follows that
  $\CSequenced \sim_{t_1} \CSequenced_1 \sim_{t_2} \CSequenced_2
  \sim_{t_3} \ldots \sim_{t_{k-1}} \CSequenced_{k-1} \sim_{t_k}
  \CSequence$, so by definition
  $\CSequenced \sim_{\{a,b\}} \CSequence$.

  Further, by the definition of the $\sim_c$ relations, we also have
  that for $i \in \{0, \LL, k-1 \}$ both
  $\CSequenced_i.\Call(\init)_a = \CSequenced_{i+1}.\Call(\init)_a$
  and
  $\CSequenced_i.\Call(\init)_b = \CSequenced_{i+1}.\Call(\init)_b$,
  where $\CSequenced_0 = \CSequenced$ and
  $\CSequenced_k = \CSequence$.  So
  $\CSequenced.\Call(\init)_a = \CSequence.\Call(\init)_a$ and
  $\CSequenced.\Call(\init)_b = \CSequence.\Call(\init)_b$.

  But $\Situation = \CSequenced(\init)$, so we get that
  $\Call(\Situation)_a = \Call(\CSequence(\init))_a$ and
  $\Call(\Situation)_b = \Call(\CSequence(\init))_b$.  
\II

\NI 
$(\supseteq)$ Take $\Situation' \in \{\Call(\Situation)\ |\ \Situation \in
\EVp{\CSequence}(a,b), \Call(\Situation)_a =
\Call(\CSequence(\init))_a \mbox{ and } \Call(\Situation)_b = \Call(\CSequence(\init))_b\}$. So for some gossip situation
$\Situation$ we have $\Situation' = \Call(\Situation)$,
$\Situation \in \EVp{\CSequence}(a,b)$,
$\Call(\Situation)_a = \Call(\CSequence(\init))_a$, and
$\Call(\Situation)_b = \Call(\CSequence(\init))_b$.  The fact that $\Situation \in \EVp{\CSequence}(a,b)$
implies that there exists a call sequence $\CSequenced$ such that
$\CSequenced \sim_{\{a,b\}} \CSequence$ and $\Situation = \CSequenced(\init)$.
Now, this and $\Call(\CSequenced(\init))_a = \Call(\Situation)_a = \Call(\CSequence(\init))_a$ imply by definition that
$\CSequenced.\Call \sim_a \CSequence.\Call$, 
so a fortiori
$\CSequenced.\Call \sim_{\{a,b\}} \CSequence.\Call$.
So
$\CSequenced.\Call(\init) \in \EVp{\CSequence.\Call}(a,b)$.
Consequently also $\Situation' \in \EVp{\CSequence.\Call}(a,b)$, because
$\Situation' = \Call(\Situation) = \CSequenced.\Call(\init)$.
\end{proof}

This allows us to conclude that
$\EVp{\CSequence.\Call}$ can be computed using $\EVp{\CSequence}$ and
$\Call$ only, i.e., without referring to $\CSequence$.  
More precisely, denote the set
of epistemic pair-views by $\EVVp$ and recall that $\Calls$ denotes the set
of calls. Then the following holds.

\begin{corollary} \label{cor:ev1} 
There exists a function
  $f: \EVVp \times \Calls \to \EVVp$ such that for any call sequence
  $\CSequence$, call $\Call$, and pair of agents $a,b \in \Agents$
\[
\EVp{\CSequence.\Call}(a,b) = f(\EVp{\CSequence}, \Call).
\]
\end{corollary}
\begin{proof}
First note that   $\EVp{\CSequence.\Call}(\gs) = \Call(\EVp{\CSequence}(\gs))$.

Suppose now that $a \neq b$.  If $\Call = ab$, then by Lemma
\ref{lem:ab-composition} $\EVp{\CSequence.\Call}(a,b)$ is a function
of $\EVp{\CSequence}(a,b)$ and $\Call$.  If $\Call \neq ab$, say
$a \not\in \Call$, then $\CSequence.\Call \sim_a \CSequence$ and hence
$\CSequence.\Call \sim_{\{a,b\}} \CSequence$, which implies
$\EVp{\CSequence.\Call}(a,b) = \EVp{\CSequence}(a,b)$.

Suppose next that $a = b$. 
By the definition of $\sim_{a}$ for all $\CSequenced$ we
have $\EVp{\CSequenced}(a,a) = \EV{\CSequenced}(a)$, so
by Corollary 2 of \cite{AW16}
$\EVp{\CSequence.\Call}(a,a)$ is a function
of $\EVp{\CSequence}(a,a)$ and $\Call$. 
\end{proof}
Consider a call sequence $\CSequence$.  If for some prefix
$\CSequence_1. \CSequence_2$ of $\CSequence$, we have
$\EVp{\CSequence_1} = \EVp{\CSequence_1.  \CSequence_2}$, then we say that
the call subsequence $\CSequence_2$ is \bfe{pair-epistemically redundant} 
in $\CSequence$ and that $\CSequence$ is \bfe{pair-epistemically redundant}.

We say that $\CSequence$ is \bfe{pair-epistemically non-redundant}
if it is not pair-epistemically redundant.
Equivalently, a call sequence
$\Call_1.\Call_2$ $.\ldots.\Call_k$ is pair-epistemically non-redundant
if the set
\[
\{\EVp{\Call_1.\Call_2.\ldots.\Call_i} \mid i \in \{1, \dots, k\}\}
\]
has $k$ elements.

\begin{lemma}[Pair-Epistemic Stuttering] \label{lem:estuttering2}
  Suppose that $\CSequence := \CSequence_1.\CSequence_2.\CSequence_3$ and
  $\CSequenced := \CSequence_1. \CSequence_3$, where $\CSequence_2$ is
  pair-epistemically redundant in $\CSequence$.
Then $\EVp{\CSequence} = \EVp{\CSequenced}$.  
\end{lemma}
\begin{proof}
  Let $\CSequence_3 = \Call_1.\Call_2.\ldots.\Call_k$.  First note
  that thanks to Corollary 2 of \cite{AW16}
we have
  $\EVp{\CSequence_1.\CSequence_2.\Call_1}$
$= \EVp{\CSequence_1.\Call_1}$, since
  $\EVp{\CSequence_1.\CSequence_2.\Call_1} =
  f(\EVp{\CSequence_1.\CSequence_2}, \Call_1) = f(\EVp{\CSequence_1},
  \Call_1) = \EVp{\CSequence_1.\Call_1}$ due to the pair-epistemic
  redundancy of $\CSequence_2$ in $\CSequence$. Repeating this
  argument for all $i \in \{1, \dots, k\}$ we get that
  $$\EVp{\CSequence_1.\CSequence_2.\Call_1.\Call_2.\ldots.\Call_i} =
  \EVp{\CSequence_1.\Call_1.\Call_2.\ldots.\Call_i}.$$
In particular $\EVp{\CSequence} = \EVp{\CSequenced}$.
\end{proof}

\begin{corollary} \label{cor:reduction2}
  For every call sequence $\CSequence$ there exists a pair-epistemically
  non-redundant call sequence $\CSequenced$ such that for all
  formulas $\phi \in {\mathcal L}_{wn}$, 
  $(\Model, \CSequence) \models \phi$ iff
  $(\Model, \CSequenced) \models \phi$.
\end{corollary}
\begin{proof}
By the repeated use of the Pair-Epistemic Stuttering Lemma  \ref{lem:estuttering2}
and Lemma \ref{lem:ck-ev}.
\end{proof}

Next, we prove the following crucial lemma.

\begin{lemma} \label{lem:ck-finite-non-epi-redundant} 
For any given model
  $\Model$, there are only finitely many pair-epistemically non-redundant
  call sequences.
\end{lemma}
\begin{proof}
  Note that each epistemic pair-view is a function from
  $\Agents \times \Agents \cup\{\gs\}$ to the set of functions from $\Agents$ to
  $2^{|\Atoms|}$ (this is an overestimation because for $\gs$ this set has only
  one element).  There are
  $k = 2^{(|\Agents|^2+1)\cdot 2^{|\Agents|\cdot|\Atoms|}}$ such functions, so any
  call sequence longer than $k$ has a pair-epistemically redundant call
  subsequence.  But there are only finitely many call sequences of
  length at most $k$.  This concludes the proof.
\end{proof}

Finally, we can establish the announced result.  

\begin{theorem}[Decidability of Truth] \label{thm:ck-edecidability} For
  any formula $\phi \in {\mathcal L}_{wn}$, it is decidable
  whether $\Model\models \phi$ holds.
\end{theorem}
\begin{proof}
Recall that
$\Model \models \phi \ \  \mbox{iff} \ \ \forall \CSequence \ (\Model, \CSequence) \models \phi.$
By Corollary \ref{cor:reduction2} we can rewrite the latter as
\[
\forall \CSequence \mbox{ s.t. } \CSequence \mbox { is pair-epistemically non-redundant, } (\Model, \CSequence) \models \phi.
\]
But according to Lemma \ref{lem:ck-finite-non-epi-redundant} there are
only finitely many pair-episte\-mi\-cally non-redundant call sequences and
by Lemma \ref{lem:ck-ev} their set can be explicitly constructed.
\end{proof}

\subsection{Decidability of termination with common knowledge operator}
\label{sec:termination-common}

Finally, we show that it is decidable to determine whether a gossip
protocol that uses guard with non-nested common knowledge operator (in
short: a common knowledge protocol) terminates.  For an example of
such a protocol see Appendix \ref{section:graphs2}.

First, we establish monotonicity of gossip situations and epistemic
pair-views with respect to call sequence extensions, w.r.t.~suitable
partial orderings.  Intuitively, we claim that as the call sequence
gets longer each agent acquires more information.

\begin{definition}
For any two gossip situations $\Situation, \Situation'$ we write
$\Situation \leqs \Situation'$ if for all $a \in \Agents$ we have
$\Situation_a \subseteq \Situation'_a$.
\end{definition}

\begin{note}[Note 1 of \cite{AW16}]\label{not:monoseq-common}
For all call sequences $\CSequence$ and $\CSequenced$ 
such that $\CSequence \sqsubseteq \CSequenced$ we have
$\CSequence(\init) \leqs \CSequenced(\init)$.
\end{note}

\begin{definition}
For any two epistemic pair-views $V, V' \in \EVVp$ we write
$V \leqtwoev V'$ if for all $a,b \in \Agents$ %
there exists $X \subseteq V(a,b)$ and an surjective (onto) function 
$g : X \to V'(a,b)$ such that for all $\Situation \in X$
we have $\Situation \leqs g(\Situation)$. 
\end{definition}

\begin{lemma}\label{lem:partial-order-common}
$\leqtwoev$ is a partial order.
\end{lemma}

\begin{proof}
\mbox{}

\NI
(Reflexivity) For any epistemic pair-view $V$, we have $V \leqtwoev V$, because for each $a,b \in \Agents$ 
we can pick $V(a,b)$ as $X$ and the identity function on $V(a,b)$ as $g$.
\II

\NI
(Transitivity) Suppose $V, V', V''$ are three epistemic pair-views such that $V \leqtwoev V'$ and $V' \leqtwoev V''$.
Then, from the definition of $\leqtwoev$,
for any $a,b \in \Agents$ there exist $X \subseteq V(a,b)$, $Y \subseteq V'(a,b)$, and surjective functions $g: X \to V'(a,b)$ and $h: Y \to V''(a,b)$. Let $Z = \{ \Situation \in X \mid g(\Situation) \in Y \}$. 
Note that $g|_Z : Z \to Y$, i.e.~the restriction of $g$ to $Z$, is surjective.
The composition $g|_Z \circ h : Z \to V''(a,b)$ is also surjective and for any gossip situation $\Situation \in Z$ the following holds
$\Situation \leqs g|_Z(\Situation) \leqs h(g|_Z(\Situation)) = (g|_Z\circ h)(\Situation)$.
\II

\NI
(Antisymmetry)  
Suppose $V, V'$ are two epistemic pair-views such that $V \leqtwoev V'$ and $V' \leqtwoev V$.
Then, from the definition of $\leqtwoev$,
for any $a,b \in \Agents$ there exist $X \subseteq V(a,b)$, $Y \subseteq V'(a,b)$, 
and surjective functions $g: X \to V'(a,b)$ and $h: Y \to V(a,b)$.
Let $Z = \{ \Situation \in X \mid g(\Situation) \in Y \}$. 
Note that $g|_Z : Z \to Y$, i.e.~the restriction of $g$ to $Z$, is surjective.
Moreover, $g|_Z \circ h : Z \to V(a,b)$ is also surjective, and because $Z \subseteq V(a,b)$ is finite,  
$Z = V(a,b)$, $g|_Z = g$, and $g \circ h$ is a permutation on $V(a,b)$. 
Similarly we can show that $Y = V'(a,b)$.
Since $(g \circ h)$ is a permutation on a finite set, 
there exists $k$ such that $(g \circ h)^k$ is the identity function on $V(a,b)$.%

Note that for any $\Situation \in V(a,b)$, we have $\Situation \leqs (g \circ h)(\Situation)$,
because $\Situation \leqs g(\Situation) \leqs h(g(\Situation)))$.
Now consider the sequence:
 $\Situation \leqs (g \circ h)(\Situation)\leqs (g \circ h)^2(\Situation)\leqs \ldots \leqs (g \circ h)^k(\Situation) = \Situation$.
In fact, all of the elements in this sequence have to be the same, because $\leqs$ is a partial order.
In particular, this shows that $(g \circ h)(\Situation) = \Situation$. 
Therefore, $g \circ h$ is the identity function on $V(a,b)$. 
Now, for any $\Situation \in V(a,b)$ we have that
$\Situation \leqs g(\Situation) \leqs h(g(\Situation)) = (g\circ h)(\Situation) = \Situation$, so $g$ is the identity function as well.
This shows that $V(a,b) = V'(a,b)$ for all $a,b \in \Agents$.
\end{proof}

The next lemma formalizes the intuition that information captured by the 
epistemic pair-view grows along a call sequence.

\begin{lemma}\label{lem:monoev-common}
For all two call sequences such that $\CSequence \sqsubseteq \CSequenced$ we have
$\EVp{\CSequence} \leqtwoev \EVp{\CSequenced}$.
\end{lemma}
\begin{proof}
Let $\CSequenced = \CSequence.\CSequence'$. Take $a,b \in \Agents$.
By a repeated application of Lemma \ref{lem:ab-composition} we get
$\EVp{\CSequence.\CSequence'}(a,b) = \{\CSequence'(\Situation)\ |\ \Situation \in \EVp{\CSequence}(a,b)$ $\mbox{ and } \forall {\CSequence'' \sqsubseteq \CSequence'}\ (\CSequence''(\Situation)_a = \CSequence''(\CSequence(\init))_a 
\wedge \CSequence''(\Situation)_b = \CSequence''(\CSequence(\init))_b) \}$. 
It suffices then to pick $X = \{ \Situation \in \EVp{\CSequence}(a,b) \mid \forall {\CSequence'' \sqsubseteq \CSequence'}\
(\CSequence''(\Situation)_a = \CSequence''(\CSequence(\init))_a \wedge \CSequence''(\Situation)_b = \CSequence''(\CSequence(\init))_b) \}$
and set $g(\Situation) = \CSequence'(\Situation)$ for all $\Situation \in X$.
It is easy to check that such $g : X \to \EVp{\CSequenced}$ is surjective,
so $\EVp{\CSequence} \leqtwoev \EVp{\CSequenced}$, as claimed.
\end{proof}

We can now draw the following useful conclusion.

\begin{lemma}\label{lem:redundant-common}
  Suppose that $\CSequence$ is pair-epistemically redundant.  Then a prefix
  $\CSequence_{1}.\Call$ of it exists such that $\CSequence_1$ is
  pair-epistemically non-redundant and
  $\EVp{\CSequence_{1}.\Call} = \EVp{\CSequence_{1}}$.
\end{lemma}
\begin{proof}
  Let $\CSequence_1.\CSequence_2$ be the shortest prefix of
  $\CSequence$ such that
  $\EVp{\CSequence_1} = \EVp{\CSequence_1.\CSequence_2}$.  Then
$\CSequence_1$   is pair-epistemically non-redundant. Let
  $\CSequence_2 = \Call_1.\ldots.\Call_l$.  By Lemma \ref{lem:monoev-common}
  we have
\[
\begin{array}{l}
  \EVp{\CSequence_1} \leqtwoev \EVp{\CSequence_1.\Call_1} \leqtwoev
  \EVp{\CSequence_1.\Call_1.\Call_2} \leqtwoev \ldots \leqtwoev \\
  \EVp{\CSequence_1.\Call_1.\Call_2.\ldots.\Call_l} =
  \EVp{\CSequence_1.\CSequence_2} = \EVp{\CSequence_1}.
\end{array}
\]  
Since $\leqtwoev$  is a partial order, $\EVp{\CSequence_1.\Call_1} = \EVp{\CSequence_1}$
  holds.
\end{proof}

Finally we can establish the desired result. In the proof
we shall use the following observation.
\begin{theorem}[Stuttering] \label{thm:ck-stuttering}
Suppose that $\CSequence := \CSequence_1. \Call. \CSequence_2$ and
$\CSequenced := \CSequence_1. \Call. \Call. \CSequence_2$.
Then for all formulas $\phi \in {\mathcal L}^{ck}$,
$(\Model, \CSequence) \models \phi$ iff $(\Model, \CSequenced) \models \phi$.
\end{theorem}
\begin{proof}
This is a direct consequence of the corresponding 
Stuttering Theorem 3 from \cite{AW16} and Note \ref{not:G}.
\end{proof}

\begin{theorem}[Decidability of Termination] \label{thm:ck-termination} 

  Given a common knowledge gossip protocol it is decidable to
  determine whether it always terminates.
\end{theorem}

\begin{proof}
  We first prove that a gossip protocol may fail to terminate iff it
  can generate a call sequence $\CSequence.\Call$ such that
  $\CSequence$ is pair-epistemically non-redundant and
  $\EVp{\CSequence.\Call} = \EVp{\CSequence}$. 

\NI 
$(\Ra)$ Let $\overline{\CSequence}$ be an infinite sequence of calls
generated by the protocol. There are only finitely many pair-epistemic
views, so some prefix $\CSequence$ of $\overline{\CSequence}$ is
pair-epistemically redundant.  The claim now follows by Lemma \ref{lem:redundant-common}.
\II

\NI
$(\La)$ Suppose that the protocol generates a sequence of calls
$\CSequence.\Call$ such that $\CSequence$ is pair-epistemically
non-redundant and $\EVp{\CSequence.\Call} = \EVp{\CSequence}$.

Let $\phi$ be the guard associated with the call $\Call$, i.e., 
$\phi \to \Call$ is a rule used in the considered protocol.
By assumption $(\Model, \CSequence) \models \phi$,
so by Lemma \ref{lem:ck-ev}
$(\Model, \CSequence.\Call) \models \phi$. 
By the repeated application of the Stuttering Theorem \ref{thm:ck-stuttering}
we get that for all $i \geq 1$, $(\Model, \CSequence.\Call^{i}) \models \phi$.
Consequently, $\CSequence.\Call^\omega$ is an infinite sequence of
calls that can be generated by the protocol. 
\II

The above equivalence shows that determining whether the protocol
always terminates is equivalent to checking that it cannot generate a
call sequence $\CSequence.\Call$ such that $\CSequence$ is
pair-epistemically non-redundant and $\EVp{\CSequence.\Call} = \EVp{\CSequence}$.

But given a call sequence, by the Decidability of Semantics Theorem
\ref{thm:ck-decidability}, it is decidable to determine whether it can be
generated by the protocol and by Lemma \ref{lem:ab-eff} it is decidable
to determine whether a call sequence is pair-epistemically
non-redundant. Further, by Lemma \ref{lem:ck-finite-non-epi-redundant}
there are only finitely many pair-epistemically non-redundant call
sequences, so the claim follows.
\end{proof}

\section{Conclusions}
\label{sec:conclusions}

We studied here various aspects of common knowledge in the context of
a natural epistemic logic used to express and reason about distributed
epistemic gossip protocols.  We showed that the semantics and truth in
this logic are decidable in the absence of nested modalities.  The
first result implies that the gossip protocols relying on such a use
of the common knowledge operator are implementable and the second one
that their partial correctness is decidable, since partial correctness
of these gossip protocols can be expressed as a formula of the
considered language.  Further, we proved that the termination of these
gossip protocols is decidable, as well.

There are a number of interesting open problems related to this work.
An obvious question is whether our results can be extended to formulas
that admit nested modalities.

In Corollary \ref{cor:ck} we showed that under certain conditions
common knowledge for two agents is equivalent to the 4th fold iterated
knowledge. An intriguing question is whether this result
holds for arbitrary call sequences and arbitrary formulas. If not,
is then common knowledge always equivalent to some finite
iterated knowledge?

Finally, it would be interesting to clarify which formulas two agents
can commonly know, i.e., given a call sequence $\CSequence$ to
characterize the formulas $\phi$ for which
$(\Model, \CSequence) \models C_{ab} \phi$ holds.  Example
\ref{exa:2agents} indicates that this problem is non-trivial even
without any call being performed.

\section*{Acknowledgments}
We thank the referees for helpful comments.  First author he was
partially supported by NCN grant 2014/13/B/ST6/01807.  The second
author was partially supported by EPSRC grants EP/M027287/1 and EP/P020909/1.

\bibliographystyle{eptcs}

\bibliography{new}

\begin{thebibliography}{10}
\providecommand{\bibitemdeclare}[2]{}
\providecommand{\surnamestart}{}
\providecommand{\surnameend}{}
\providecommand{\urlprefix}{Available at }
\providecommand{\url}[1]{\texttt{#1}}
\providecommand{\href}[2]{\texttt{#2}}
\providecommand{\urlalt}[2]{\href{#1}{#2}}
\providecommand{\doi}[1]{doi:\urlalt{http://dx.doi.org/#1}{#1}}
\providecommand{\bibinfo}[2]{#2}

\bibitemdeclare{inproceedings}{AGH16}
\bibitem{AGH16}
\bibinfo{author}{K.~R. \surnamestart Apt\surnameend},
  \bibinfo{author}{D.~\surnamestart Grossi\surnameend} \&
  \bibinfo{author}{W.~\surnamestart {van der Hoek}\surnameend}
  (\bibinfo{year}{2016}): \emph{\bibinfo{title}{Epistemic Protocols for
  Distributed Gossiping}}.
\newblock In: {\sl \bibinfo{booktitle}{Proceedings of the 15th Conference on
  Theoretical Aspects of Rationality and Knowledge (TARK 2015)}}, {\sl
  \bibinfo{series}{{EPTCS}}} \bibinfo{volume}{215}, pp.
  \bibinfo{pages}{51--66}, \doi{10.4204/EPTCS.215.5}.

\bibitemdeclare{inproceedings}{AW16}
\bibitem{AW16}
\bibinfo{author}{K.~R. \surnamestart Apt\surnameend} \&
  \bibinfo{author}{D.~\surnamestart Wojtczak\surnameend}
  (\bibinfo{year}{2016}): \emph{\bibinfo{title}{On Decidability of a Logic of
  Gossips}}.
\newblock In: {\sl \bibinfo{booktitle}{Proceedings of the 15th European
  Conference, {JELIA} 2016}}, {\sl \bibinfo{series}{Lecture Notes in Computer
  Science}} \bibinfo{volume}{10021}, \bibinfo{publisher}{Springer}, pp.
  \bibinfo{pages}{18--33}, \doi{10.1007/978-3-319-48758-8_2}.

\bibitemdeclare{inproceedings}{AKW17}
\bibitem{AKW17}
\bibinfo{author}{K.R. \surnamestart Apt\surnameend},
  \bibinfo{author}{E.~\surnamestart Kopczy\'{n}ski\surnameend} \&
  \bibinfo{author}{D.~\surnamestart Wojtczak\surnameend}
  (\bibinfo{year}{2017}): \emph{\bibinfo{title}{On the Computational Complexity
  of Gossip Protocols}}.
\newblock In: {\sl \bibinfo{booktitle}{Proceedings of 26th IJCAI}}.
\newblock \bibinfo{note}{To appear.}

\bibitemdeclare{inproceedings}{AW17}
\bibitem{AW17}
\bibinfo{author}{K.R. \surnamestart Apt\surnameend} \&
  \bibinfo{author}{D.~\surnamestart Wojtczak\surnameend}
  (\bibinfo{year}{2017}): \emph{\bibinfo{title}{{Decidability of Fair
  Termination of Gossip Protocols}}}.
\newblock In: {\sl \bibinfo{booktitle}{Proceedings of the {IWIL Workshop and
  LPAR Short Presentations}}}, \bibinfo{publisher}{{Kalpa Publications}}, pp.
  \bibinfo{pages}{73--85}.

\bibitemdeclare{inproceedings}{ADGH14a}
\bibitem{ADGH14a}
\bibinfo{author}{M.~\surnamestart Attamah\surnameend},
  \bibinfo{author}{H.~\surnamestart {van Ditmarsch}\surnameend},
  \bibinfo{author}{D.~\surnamestart Grossi\surnameend} \&
  \bibinfo{author}{W.~\surnamestart {van der Hoek}\surnameend}
  (\bibinfo{year}{2014}): \emph{\bibinfo{title}{A Framework for Epistemic
  Gossip Protocols}}.
\newblock In: {\sl \bibinfo{booktitle}{Proceedings of the 12th European
  Conference on Multi-Agent Systems (EUMAS 2014), Revised Selected Papers}},
  \bibinfo{volume}{8953}, \bibinfo{publisher}{Springer}, pp.
  \bibinfo{pages}{193--209}, \doi{10.1007/978-3-319-17130-2_13}.

\bibitemdeclare{inproceedings}{ADGH14}
\bibitem{ADGH14}
\bibinfo{author}{M.~\surnamestart Attamah\surnameend},
  \bibinfo{author}{H.~\surnamestart {van Ditmarsch}\surnameend},
  \bibinfo{author}{D.~\surnamestart Grossi\surnameend} \&
  \bibinfo{author}{W.~\surnamestart {van der Hoek}\surnameend}
  (\bibinfo{year}{2014}): \emph{\bibinfo{title}{Knowledge and Gossip}}.
\newblock In: {\sl \bibinfo{booktitle}{Proceedings of ECAI'14}},
  \bibinfo{publisher}{IOS Press}, pp. \bibinfo{pages}{21--26},
  \doi{10.3233/978-1-61499-419-0-21}.

\bibitemdeclare{inproceedings}{BenEijKoo05:ckiul}
\bibitem{BenEijKoo05:ckiul}
\bibinfo{author}{J.~van \surnamestart Benthem\surnameend},
  \bibinfo{author}{J.~\surnamestart van Eijck\surnameend} \&
  \bibinfo{author}{B.~\surnamestart Kooi\surnameend} (\bibinfo{year}{2005}):
  \emph{\bibinfo{title}{Common Knowledge in Update Logics}}.
\newblock In: {\sl \bibinfo{booktitle}{Proceedings of the 10th Conference on
  Theoretical Aspects of Rationality and Knowledge (TARK 2005)}}, pp.
  \bibinfo{pages}{253--261}, \doi{10.1145/1089933.1089960}.

\bibitemdeclare{inproceedings}{CHMMR16}
\bibitem{CHMMR16}
\bibinfo{author}{M.C. \surnamestart Cooper\surnameend},
  \bibinfo{author}{A.~\surnamestart Herzig\surnameend},
  \bibinfo{author}{F.~\surnamestart Maffre\surnameend},
  \bibinfo{author}{F.~\surnamestart Maris\surnameend} \&
  \bibinfo{author}{P.~\surnamestart R{\'{e}}gnier\surnameend}
  (\bibinfo{year}{2016}): \emph{\bibinfo{title}{A simple account of multiagent
  epistemic planning}}.
\newblock In: {\sl \bibinfo{booktitle}{Proceedings of {ECAI} 2016}},
  \bibinfo{publisher}{IOS Press}, pp. \bibinfo{pages}{193--201},
  \doi{10.3233/978-1-61499-672-9-193}.

\bibitemdeclare{inproceedings}{cooper_simple_2016}
\bibitem{cooper_simple_2016}
\bibinfo{author}{M.C. \surnamestart Cooper\surnameend},
  \bibinfo{author}{A.~\surnamestart Herzig\surnameend},
  \bibinfo{author}{F.~\surnamestart Maffre\surnameend},
  \bibinfo{author}{F.~\surnamestart Maris\surnameend} \&
  \bibinfo{author}{P.~\surnamestart Regnier\surnameend} (\bibinfo{year}{2016}):
  \emph{\bibinfo{title}{Simple {{Epistemic Planning}}: {{Generalised
  Gossiping}}}}.
\newblock In: {\sl \bibinfo{booktitle}{Proceedings of {ECAI} 2016}}, {\sl
  \bibinfo{series}{Frontiers in Artificial Intelligence and Applications}}
  \bibinfo{volume}{285}, \bibinfo{publisher}{{IOS Press}}, pp.
  \bibinfo{pages}{1563--1564}, \doi{10.3233/978-1-61499-672-9-1563}.

\bibitemdeclare{article}{DEPRS17}
\bibitem{DEPRS17}
\bibinfo{author}{H.~\surnamestart van Ditmarsch\surnameend},
  \bibinfo{author}{J.~\surnamestart van Eijck\surnameend},
  \bibinfo{author}{P.~\surnamestart Pardo\surnameend},
  \bibinfo{author}{R.~\surnamestart Ramezanian\surnameend} \&
  \bibinfo{author}{F.~\surnamestart Schwarzentruber\surnameend}
  (\bibinfo{year}{2017}): \emph{\bibinfo{title}{Epistemic Protocols for Dynamic
  Gossip}}.
\newblock {\sl \bibinfo{journal}{J. of Applied Logic}}
  \bibinfo{volume}{20}(\bibinfo{number}{C}), pp. \bibinfo{pages}{1--31},
  \doi{10.1016/j.jal.2016.12.001}.

\bibitemdeclare{incollection}{DEV09}
\bibitem{DEV09}
\bibinfo{author}{H.~\surnamestart van Ditmarsch\surnameend},
  \bibinfo{author}{J.~\surnamestart van Eijck\surnameend} \&
  \bibinfo{author}{R.~\surnamestart Verbrugge\surnameend}
  (\bibinfo{year}{2009}): \emph{\bibinfo{title}{Common Knowledge and Common
  Belief}}.
\newblock In \bibinfo{editor}{J.~\surnamestart van Eijck\surnameend} \&
  \bibinfo{editor}{R.~\surnamestart Verbrugge\surnameend}, editors: {\sl
  \bibinfo{booktitle}{Discourses on Social Software}},
  \bibinfo{publisher}{Amsterdam University Press}, pp.
  \bibinfo{pages}{99--122}.

\bibitemdeclare{inproceedings}{van_ditmarsch_parameters_2016}
\bibitem{van_ditmarsch_parameters_2016}
\bibinfo{author}{H.~\surnamestart van Ditmarsch\surnameend},
  \bibinfo{author}{D.~\surnamestart Grossi\surnameend},
  \bibinfo{author}{A.~\surnamestart Herzig\surnameend},
  \bibinfo{author}{W.~\surnamestart {van der Hoek}\surnameend} \&
  \bibinfo{author}{L.B. \surnamestart Kuijer\surnameend}
  (\bibinfo{year}{2016}): \emph{\bibinfo{title}{Parameters for Epistemic Gossip
  Problems}}.
\newblock In: {\sl \bibinfo{booktitle}{Proceedings of the 12th Conference on
  Logic and the Foundations of Game and Decision Theory (LOFT 2016)}}.
\newblock \bibinfo{note}{Available at
  \url{https://pdfs.semanticscholar.org/74b5/2c025f335ba487cac612019e39ce6c818448.pdf}}.

\bibitemdeclare{book}{hvdetal.del:2007}
\bibitem{hvdetal.del:2007}
\bibinfo{author}{H.~\surnamestart van Ditmarsch\surnameend},
  \bibinfo{author}{W.~\surnamestart {van der Hoek}\surnameend} \&
  \bibinfo{author}{B.~\surnamestart Kooi\surnameend} (\bibinfo{year}{2007}):
  \emph{\bibinfo{title}{Dynamic Epistemic Logic}}.
\newblock {\sl \bibinfo{series}{Synthese Library}} \bibinfo{volume}{337},
  \bibinfo{publisher}{Springer}, \doi{10.1007/978-1-4020-5839-4}.

\bibitemdeclare{book}{FHMV_RAK}
\bibitem{FHMV_RAK}
\bibinfo{author}{R.~\surnamestart Fagin\surnameend},
  \bibinfo{author}{J.~\surnamestart Halpern\surnameend},
  \bibinfo{author}{M.~\surnamestart Vardi\surnameend} \&
  \bibinfo{author}{Y.~\surnamestart Moses\surnameend} (\bibinfo{year}{1995}):
  \emph{\bibinfo{title}{Reasoning about knowledge}}.
\newblock \bibinfo{publisher}{MIT Press}, \bibinfo{address}{Cambridge,
  Massachusetts}.

\bibitemdeclare{article}{Friedell}
\bibitem{Friedell}
\bibinfo{author}{M.F. \surnamestart Friedell\surnameend}
  (\bibinfo{year}{1969}): \emph{\bibinfo{title}{On the structure of shared
  awareness}}.
\newblock {\sl \bibinfo{journal}{Behavioral Science}}
  \bibinfo{volume}{14}(\bibinfo{number}{1}), pp. \bibinfo{pages}{28--39},
  \doi{10.1002/bs.3830140105}.

\bibitemdeclare{article}{HHL88}
\bibitem{HHL88}
\bibinfo{author}{S.M. \surnamestart Hedetniemi\surnameend},
  \bibinfo{author}{S.T. \surnamestart Hedetniemi\surnameend} \&
  \bibinfo{author}{A.L. \surnamestart Liestman\surnameend}
  (\bibinfo{year}{1988}): \emph{\bibinfo{title}{A survey of gossiping and
  broadcasting in communication networks}}.
\newblock {\sl \bibinfo{journal}{Networks}}
  \bibinfo{volume}{18}(\bibinfo{number}{4}), pp. \bibinfo{pages}{319--349},
  \doi{10.1002/net.3230180406}.

\bibitemdeclare{article}{herzig_how_2017}
\bibitem{herzig_how_2017}
\bibinfo{author}{A.~\surnamestart Herzig\surnameend} \&
  \bibinfo{author}{F.~\surnamestart Maffre\surnameend} (\bibinfo{year}{2017}):
  \emph{\bibinfo{title}{How to Share Knowledge by Gossiping}}.
\newblock {\sl \bibinfo{journal}{AI Communications}}
  \bibinfo{volume}{30}(\bibinfo{number}{1}), pp. \bibinfo{pages}{1--17},
  \doi{10.3233/AIC-170723}.
\newblock
  \urlprefix\url{http://content.iospress.com/articles/ai-communications/aic723}.

\bibitemdeclare{book}{HKPRU05}
\bibitem{HKPRU05}
\bibinfo{author}{J.~\surnamestart Hromkovic\surnameend},
  \bibinfo{author}{R.~\surnamestart Klasing\surnameend},
  \bibinfo{author}{A.~\surnamestart Pelc\surnameend},
  \bibinfo{author}{P.~\surnamestart Ruzicka\surnameend} \&
  \bibinfo{author}{W.~\surnamestart Unger\surnameend} (\bibinfo{year}{2005}):
  \emph{\bibinfo{title}{Dissemination of Information in Communication Networks
  - Broadcasting, Gossiping, Leader Election, and Fault-Tolerance}}.
\newblock \bibinfo{series}{Texts in Theoretical Computer Science. An {EATCS}
  Series}, \bibinfo{publisher}{Springer}, \doi{10.1007/b137871}.

\bibitemdeclare{book}{Lewis:1969}
\bibitem{Lewis:1969}
\bibinfo{author}{D.K. \surnamestart Lewis\surnameend} (\bibinfo{year}{1969}):
  \emph{\bibinfo{title}{Convention, a Philosophical Study}}.
\newblock \bibinfo{publisher}{Harvard University Press},
  \bibinfo{address}{Cambridge (MA)}.

\bibitemdeclare{inproceedings}{VS14}
\bibitem{VS14}
\bibinfo{author}{P.~\surnamestart Vanderschraaf\surnameend} \&
  \bibinfo{author}{G.~\surnamestart Sillari\surnameend} (\bibinfo{year}{2014}):
  \emph{\bibinfo{title}{Common Knowledge}}.
\newblock In: {\sl \bibinfo{booktitle}{The Stanford Encyclopedia of
  Philosophy}}.
\newblock \bibinfo{note}{Available at
  \url{https://plato.stanford.edu/entries/common-knowledge}}.

\bibitemdeclare{inproceedings}{wangetal.tark:2009}
\bibitem{wangetal.tark:2009}
\bibinfo{author}{Y.~\surnamestart Wang\surnameend},
  \bibinfo{author}{L.~\surnamestart Kuppusamy\surnameend} \&
  \bibinfo{author}{J.~\surnamestart van Eijck\surnameend}
  (\bibinfo{year}{2009}): \emph{\bibinfo{title}{Verifying epistemic protocols
  under common knowledge}}.
\newblock In: {\sl \bibinfo{booktitle}{Proceedings of the 12th Conference on
  Theoretical Aspects of Rationality and Knowledge (TARK 2009)}},
  \bibinfo{publisher}{ACM}, pp. \bibinfo{pages}{257--266},
  \doi{10.1145/1562814.1562848}.

\end{thebibliography}

\appendix

\section{Example: a common knowledge protocol}
\label{section:graphs2}

To illustrate gossip protocols that employ the common knowledge
operator assume  that the agents are
nodes of an undirected connected graph $(V,E)$ and that the calls can
take place only between pairs of agents connected by an edge. 
Let $N_i$ denote the set of neighbours of node $i$. 

Consider a gossip protocol with the following program for agent $i$
(we use here the syntax introduced in  \cite{AGH16}):
\[
*[[]_{j \in N_i, B \in \mathsf{P}} F_i B \wedge \neg C_{ij} F_j B \to (i, j)].
\]

Informally, agent $i$ calls a neighbour $j$ if $i$ is familiar with
some secret (here $B$) and there is no common knowledge between $i$
and $j$ that $j$ is familiar with this secret.

\bfe{Partial correctness} of a protocol states that upon its termination the formula
$\bigwedge_{i, j \in \mathsf{A}} F_i J$ holds.

To prove partial correctness of the above protocol consider the exit condition 
\[
\bigwedge_{(i,j) \in E} \bigwedge_{B \in \mathsf{P}} (F_i B \to C_{ij} F_j B).
\]
For all agents $i$ and $j$ and
secrets $B$, the formula $C_{ij} F_j B \to F_j B$ is true,
so the exit condition
implies
\[
\bigwedge_{(i,j) \in E} \bigwedge_{B \in \mathsf{P}} (F_i B \to F_j B).
\]

Consider now an agent $i$ and the secret $J$ of agent $j$. Let
$j = i_1, \LL, i_h = i$ be a path that connects $j$ with $i$. The
above formula implies that for $g \in \{1, \LL, h-1\}$ we have
$\bigwedge_{B \in \mathsf{P}} (F_{i_g} B \to F_{i_{g+1}} B)$. By
combining these $h-1$ formulas we get
$\bigwedge_{B \in \mathsf{P}} (F_{j} B \to F_{i} B)$. But $F_j J$ is
true, so we conclude $F_i J$.  Consequently
$\bigwedge_{i, j \in \mathsf{A}} F_i J$, as desired.

To prove termination we need the following observation.

\begin{lemma}
\label{lem:ij}
For all call sequences $\CSequence.(i,j)$ and secrets $B$
\[
\mbox{$(\Model, \CSequence.(i,j)) \models F_i B \land F_j B$ implies
$(\Model, \CSequence.(i,j)) \models C_{ij} (F_i B \land F_j B)$.}
\]
\end{lemma}
\begin{proof}
  Suppose that $(\Model, \CSequence.(i,j)) \models F_i B \land F_j B$.
  Take some $\CSequenced$ such that
  $\CSequence.(i,j) \sim_{\{i,j\}} \CSequenced$. So there exists a
  sequence of call sequences $\CSequence_1, \LL, \CSequence_k$ such
  that
\[
\CSequence.(i,j) = \CSequence_1 \sim_{m_1} \CSequence_2 \sim_{m_2} \LL \sim_{m_k} \CSequence_k = \CSequenced,
\] 
where each $m_h$ is $i$ or $j$.

  By the repeated use of the Equivalence Theorem \ref{thm:equiv} each of the
  call sequences $\CSequence_1, \LL, \CSequence_k$ contains the call
  $(i,j)$ that corresponds with the last call of $\CSequence.(i,j)$.
  Let $\CSequenced_1.(i,j), \LL, \CSequenced_k.(i,j)$ be the
  corresponding prefixes of $\CSequence_1, \LL, \CSequence_k$.

  By the Equivalence Theorem  \ref{thm:equiv}
and the definition of a view given there we also have
\[
\CSequenced_1.(i,j)_{i} = \CSequenced_1.(i,j)_{j} = 
\CSequenced_2.(i,j)_{i} = \CSequenced_2.(i,j)_{j}.
\]
By the assumption this implies $(\Model, \CSequenced_2.(i,j)) \models F_i B \land F_j B$,
since $\CSequenced_1.(i,j) = \CSequence.(i,j)$.

Repeating this procedure we conclude that
$(\Model, \CSequenced_k.(i,j)) \models F_i B \land F_j B$
and hence $(\Model, \CSequenced) \models F_i B \land F_j B$.
This implies the claim.
\end{proof}

Now, by the definition of semantics for all call sequences
$\CSequence.(i,j)$ and secrets $B$,
$(\Model, \CSequence) \models F_i B$ implies
$(\Model, \CSequence.(i,j)) \models F_i B \land F_j B$, which implies
by Lemma \ref{lem:ij} that
$(\Model, \CSequence.(i,j)) \models C_{ij} (F_i B \land F_j B)$ and hence
$(\Model, \CSequence.(i,j)) \models C_{ij} F_j B$.  This
shows that after each call $(i, j)$ the size of the set
$\{(i,j,B) \mid \neg C_{ij} F_j B\}$ decreases.

\end{document}
\section{Omitted Proofs}
\NI
\textbf{Proof of Lemma \ref{lem:Ra}}.

\begin{proof}
  Since $\CSequence$ does not contain the call $ab$, the views of $a$
  of $\CSequence$ and $R_{ab}(\CSequence)$ have the
  same sequences of the $a$-calls.

  By definition agent $a$ does not learn the secret of $b$ through the
  $b$-inessential calls for $a$. So the replacement (or possibly
  deletion) of $b$ in all $b$-inessential calls has no effect on the
  status of this secret in the views of $a$ of both sequences.

  Let $bc$ be the first $b$-inessential for $a$ call in $\CSequence$
  and suppose that $ad$ is the first $a$-call in $\CSequence$ it leads
  to.  Let $\CSequence'$ be the outcome of the first step of producing
  $R_{ab}(\CSequence)$.  So it is obtained from $\CSequence$ by
  replacing $bc$ by $cd$ if $c \neq d$, and by deleting $bc$
  otherwise.  We show that the views of the agent $a$ of the call
  sequences $\CSequence$ and $\CSequence'$ are the same.  First, note
  that no $a$-call earlier than $ad$ can be effected by the change in
  $\CSequence$, because $ad$ is the first $a$-call $bc$ leads to.
  
  Consider the $c \neq d$ case first.  The set of secrets an agent is
  familiar with at the same point in $\CSequence$ and $\CSequence'$
  may differ as a result of $bc$ being replaced by $cd$.  However, we
  argue that it is impossible for the agent $a$ to notice this
  difference.  Let us consider a call $ax$ in $\CSequence$, where
  $x \in \Agents$ and the sets of secrets, $S$ and $S'$, agent $e$ is
  familiar before $ax$ is made in $\CSequence$ and $\CSequence'$,
  respectively.
  \begin{itemize}
  \item If $ax$ takes place before $bc$ then $S$ and $S'$ are the same.
  \item If $ax$ is in-between the calls $bc$ and $ad$
  then again the sets $S$ and $S'$ are the same, because 
  $ad$ is the first $a$-call that $bc$ leads to.
  \item If $ax$ is the call $ad$ then we have the following.
  First, just before this call $a$ is already familiar with all the secrets $b$ 
  is familiar with before the call $bc$ is made in $\CSequence$,
  because $ad$ is the first $a$-call $bc$ leads to.
  So $a$ is still familiar with all these secrets in $\CSequence'$. Thus
 the only secrets that may be lost by replacing $bc$ by $cd$ 
  are the ones that $c$ is familiar with at that point.
  However, these secrets are passed to $d$ and $a$ 
  still learns them all in $\CSequence'$ through the call $ad$.
\item If $ax$ takes place after the call $ad$ then the difference
  between $S$ and $S'$ could be at most in the set of secrets $a$
  learned through the call $bc$.  However, $a$ already knows these
  secrets after the call $ad$ is made, so $S = S'$.\footnote{Note
    that this crucially depends on the fact that from any call each
    involved agent learns the union of the sets of secrets the callers
    are familiar with and not the set of secrets the other caller is
    familiar with.}
  \end{itemize}
The reasoning for the $c = d$ case is completely analogous and omitted.

By iterating the above argument, starting with
$\CSequence$, we obtain $R_{ab}(\CSequence)$ without affecting the
view of the agent $a$.
\end{proof}

\III

\NI
\textbf{Proof of Corollary \ref{cor:aba}}.

\begin{proof}
Take $\CSequenced$ such that $\CSequence \sim_a
\CSequenced$. Then $\CSequenced$ is of the form
$\CSequence_1, ac, \CSequence_2, bc,$ $\CSequence_3, ac, \CSequence_4$. Next take
$\CSequenced_1$ such that $\CSequenced_1 \sim_b \CSequenced$. Then
$\CSequenced$ is of the form $\CSequence_5, ac, \CSequence_6, bc, \CSequence_7$. Next take
$\CSequenced_2$ such that $\CSequenced_2 \sim_a \CSequenced_1$. Then
$\CSequenced_2$ is of the form $\CSequence_8, ac, \CSequence_9$.  So
$(\Model, \CSequenced_2) \models F_c A$, which implies the claim.

The next two claims follow since for all formulas
$\phi \in {\mathcal L}^{ck}$, $K_{aba} \phi$ implies both
$K_{ab} \phi$ and $K_{a} \phi$.

Now note that for all sequences $t \in \{a, b\}^*$ that extend $abab$
or $baba$ and all formulas $\phi \in {\mathcal L}^{ck}$, $K_t \phi$
implies $K_{abab} \phi$ or $K_{baba} \phi$.  So Theorem
\ref{thm:abab1} implies that for all such sequences $s$ and all call
sequences $\CSequence$ that do not contain the call $ab$
\[
(\Model, \CSequence) \not \models K_{s} F_c A.
\]
Indeed, the formula $F_c A$ is not always true.

Further, by Note \ref{not:G} we also have for such call sequences
$\CSequence$
\[
(\Model, \CSequence) \not \models C_{ab} F_c A.
\]
In particular, this holds for the above call sequence
$\CSequence = ac, bc, ac$. 
\end{proof}

\III

\NI
\textbf{Proof of Lemma \ref{lem:ab-finite}}.

\begin{proof}
  Consider an $ab$-redundant free call sequence $\CSequenced$ such that
  $\CSequence \sim_{\{a, b\}} \CSequenced$. Then $\CSequenced$ has the same
  number, say $k$, of calls $ab$ as $\CSequence$.
  
  Associate with $\CSequenced$ the sequence of gossip situations
  $$\CSequenced^0(\init), \CSequenced^1(\init), \ldots,
  \CSequenced^m(\init),$$ 
where $m$ is the length of $\CSequenced$,
  $\CSequenced^0 = \epsilon$, and
  $\CSequenced^k = \Calld_1, \Calld_2, \ldots, \Calld_k$ for
  $k = 1, \dots, m$. This sequence monotonically grows, where we
  interpret the inclusion relation componentwise. Moreover, for all
  calls $\Calld_i$ different from $ab$ the corresponding
  inclusion is strict.  Consequently, $m$, the length of
  $\CSequenced$, is bounded by $k + |\Agents|^{|\Agents|}$, the sum of
  the number of calls $ab$ in $\CSequence$ and of the total number of
  secrets in the gossip situation in which each agent is an expert.

But for each $m$ there are only finitely many call sequences
of length at most $m$.  This concludes the proof.
\end{proof}

\III

\NI
\textbf{Proof of Lemma \ref{lem:ck-ev}}.

\begin{proof}
A straightforward proof by induction shows that for a formula $\psi \in {\mathcal L}_{pr}$ and
arbitrary call sequences $\CSequence'$ and $\CSequenced'$,
\begin{equation}
\mbox{$\CSequence'(\init) = \CSequenced'(\init)$ implies that
  $(\Model, \CSequence') \models \psi$ iff $(\Model, \CSequenced') \models \psi$.}
  \label{equ:2}  
\end{equation}
Since $\EVp{\CSequence}(\gs) = \CSequence(\init)$ and 
$\EVp{\CSequenced}(\gs) = \CSequenced(\init)$, this settles the case for
$\phi = F_a p$.

Next, consider the case of the formulas of the form $C_G \phi$, where
$\phi \in \mathcal L_{wn}$.  

If $|G| = 1$, then $G = \{a\}$ for some $a \in \Agents$ and $C_G$ is
the same as $K_a$.  Since $\EVp{\CSequence} = \EVp{\CSequenced}$
implies $\EV{\CSequence} = \EV{\CSequenced}$, where epistemic view $\EV{}$ is defined as in \cite{AW16},
the claim follows by Lemma 4 of \cite{AW16}.

If $|G| = 2$, then $G = \{a,b\}$ for some $a,b \in \Agents$.
By (\ref{equ:2}) and the definition of $\EVp{\CSequence}$
$$(\Model, \CSequence) \models C_G \phi  \mbox{ iff }
\forall \CSequence' \mbox{  s.t.  } \CSequence'(\init) \in \EVp{\CSequence}(a,b)
, ~(\Model, \CSequence') \models \phi.$$
So the claim follows since $\EVp{\CSequence}(a,b) = \EVp{\CSequenced}(a,b)$.

If $|G| \geq 3$, then by Theorem \ref{thm:3agents} 
both 
$(\Model, \CSequence') \models C_G \phi$ and $(\Model, \CSequenced') \models C_G \phi$
are equivalent to $\models \phi$.

This settles the case for $C_G \phi$.  The remaining cases of
negation and conjunction follow directly by the induction.
\end{proof}

\III

\NI
\textbf{Proof of Lemma \ref{lem:ab-composition}}.

\begin{proof}
\II

\NI
$(\sse)$ Take $\Situation' \in \EVp{\CSequence.\Call}(a,b)$. By the
  definition of $\EVp{\CSequence.\Call}(a,b)$ there exists a call
  sequence $\CSequenced$ such that
  $\CSequenced.\Call \sim_{\{a,b\}} \CSequence.\Call$ and
  $\Situation' = \CSequenced.\Call(\init)$. 
  So
  $\Situation' = \Call(\Situation)$, where
  $\Situation = \CSequenced(\init)$.
  We prove now that
  $\CSequenced \sim_{\{a,b\}} \CSequence$
  and as a result $\Situation = \CSequenced(\init) \in \EVp{\CSequence}(a,b)$. 

  Note that $\CSequenced.\Call \sim_{\{a,b\}} \CSequence.\Call$
  implies that there exists a 
  sequence $t_1 \ldots t_k \in \{a,b\}^*$ such that
  for some call sequences 
  $\CSequenced_1, \CSequenced_2, \ldots, \CSequenced_{k-1}$
we have
  $\CSequenced.\Call \sim_{t_1} \CSequenced_1.\Call \sim_{t_2} \CSequenced_2.\Call \sim_{t_3} \ldots
  \sim_{t_{k-1}} \CSequenced_{k-1}.\Call \sim_{t_k} \CSequence.\Call$. 
  Note that for any $t \in \{a,b\}$ and call sequences
  $\CSequence', \CSequenced'$ we have that
  $\CSequence'.\Call \sim_{t} \CSequenced'.\Call$ implies
  $\CSequence' \sim_{t} \CSequenced'$, because $\sim_{t}$ is the
  minimal relation satisfying the conditions stated in Definition
  \ref{def:model}.  It follows that
  $\CSequenced \sim_{t_1} \CSequenced_1 \sim_{t_2} \CSequenced_2
  \sim_{t_3} \ldots \sim_{t_{k-1}} \CSequenced_{k-1} \sim_{t_k}
  \CSequence$, so by definition
  $\CSequenced \sim_{\{a,b\}} \CSequence$.

  Further, by the definition of the $\sim_c$ relations, we also have
  that for $i \in \{0, \LL, k-1 \}$ both
  $\CSequenced_i.\Call(\init)_a = \CSequenced_{i+1}.\Call(\init)_a$
  and
  $\CSequenced_i.\Call(\init)_b = \CSequenced_{i+1}.\Call(\init)_b$,
  where $\CSequenced_0 = \CSequenced$ and
  $\CSequenced_k = \CSequence$.  So
  $\CSequenced.\Call(\init)_a = \CSequence.\Call(\init)_a$ and
  $\CSequenced.\Call(\init)_b = \CSequence.\Call(\init)_b$.

  But $\Situation = \CSequenced(\init)$, so we get that
  $\Call(\Situation)_a = \Call(\CSequence(\init))_a$ and
  $\Call(\Situation)_b = \Call(\CSequence(\init))_b$.  
\II

\NI 
$(\supseteq)$ Take $\Situation' \in \{\Call(\Situation)\ |\ \Situation \in
\EVp{\CSequence}(a,b), \Call(\Situation)_a =
\Call(\CSequence(\init))_a \mbox{ and } \Call(\Situation)_b = \Call(\CSequence(\init))_b\}$. So for some gossip situation
$\Situation$ we have $\Situation' = \Call(\Situation)$,
$\Situation \in \EVp{\CSequence}(a,b)$,
$\Call(\Situation)_a = \Call(\CSequence(\init))_a$, and
$\Call(\Situation)_b = \Call(\CSequence(\init))_b$.  The fact that $\Situation \in \EVp{\CSequence}(a,b)$
implies that there exists a call sequence $\CSequenced$ such that
$\CSequenced \sim_{\{a,b\}} \CSequence$ and $\Situation = \CSequenced(\init)$.
Now, this and $\Call(\CSequenced(\init))_a = \Call(\Situation)_a = \Call(\CSequence(\init))_a$ imply by definition that
$\CSequenced.\Call \sim_a \CSequence.\Call$, 
so a fortiori
$\CSequenced.\Call \sim_{\{a,b\}} \CSequence.\Call$.
So
$\CSequenced.\Call(\init) \in \EVp{\CSequence.\Call}(a,b)$.
Consequently also $\Situation' \in \EVp{\CSequence.\Call}(a,b)$, because
$\Situation' = \Call(\Situation) = \CSequenced.\Call(\init)$.
\end{proof}
\III

\NI
\textbf{Proof of Lemma \ref{lem:partial-order-common}}.

\begin{proof}
\mbox{}

\NI
(Reflexivity) For any epistemic pair-view $V$, we have $V \leqtwoev V$, because for each $a,b \in \Agents$ 
we can pick $V(a,b)$ as $X$ and the identity function on $V(a,b)$ as $g$.
\II

\NI
(Transitivity) Suppose $V, V', V''$ are three epistemic pair-views such that $V \leqtwoev V'$ and $V' \leqtwoev V''$.
Then, from the definition of $\leqtwoev$,
for any $a,b \in \Agents$ there exist $X \subseteq V(a,b)$, $Y \subseteq V'(a,b)$, and surjective functions $g: X \to V'(a,b)$ and $h: Y \to V''(a,b)$. Let $Z = \{ \Situation \in X \mid g(\Situation) \in Y \}$. 
Note that $g|_Z : Z \to Y$, i.e.~the restriction of $g$ to $Z$, is surjective.
The composition $g|_Z \circ h : Z \to V''(a,b)$ is also surjective and for any gossip situation $\Situation \in Z$ the following holds
$\Situation \leqs g|_Z(\Situation) \leqs h(g|_Z(\Situation)) = (g|_Z\circ h)(\Situation)$.
\II

\NI
(Antisymmetry)  
Suppose $V, V'$ are two epistemic pair-views such that $V \leqtwoev V'$ and $V' \leqtwoev V$.
Then, from the definition of $\leqtwoev$,
for any $a,b \in \Agents$ there exist $X \subseteq V(a,b)$, $Y \subseteq V'(a,b)$, 
and surjective functions $g: X \to V'(a,b)$ and $h: Y \to V(a,b)$.
Let $Z = \{ \Situation \in X \mid g(\Situation) \in Y \}$. 
Note that $g|_Z : Z \to Y$, i.e.~the restriction of $g$ to $Z$, is surjective.
Moreover, $g|_Z \circ h : Z \to V(a,b)$ is also surjective, and because $Z \subseteq V(a,b)$ is finite,  
$Z = V(a,b)$, $g|_Z = g$, and $g \circ h$ is a permutation on $V(a,b)$. 
Similarly we can show that $Y = V'(a,b)$.
Since $(g \circ h)$ is a permutation on a finite set, 
there exists $k$ such that $(g \circ h)^k$ is the identity function on $V(a,b)$.%

Note that for any $\Situation \in V(a,b)$, we have $\Situation \leqs (g \circ h)(\Situation)$,
because $\Situation \leqs g(\Situation) \leqs h(g(\Situation)))$.
Now consider the sequence:
 $\Situation \leqs (g \circ h)(\Situation)\leqs (g \circ h)^2(\Situation)\leqs \ldots \leqs (g \circ h)^k(\Situation) = \Situation$.
In fact, all of the elements in this sequence have to be the same, because $\leqs$ is a partial order.
In particular, this shows that $(g \circ h)(\Situation) = \Situation$. 
Therefore, $g \circ h$ is the identity function on $V(a,b)$. 
Now, for any $\Situation \in V(a,b)$ we have that
$\Situation \leqs g(\Situation) \leqs h(g(\Situation)) = (g\circ h)(\Situation) = \Situation$, so $g$ is the identity function as well.
This shows that $V(a,b) = V'(a,b)$ for all $a,b \in \Agents$.
\end{proof}
\III

\NI
\textbf{Proof of the Decidability of
  Termination Theorem \ref{thm:ck-termination}}.

We shall use the following observation.
\begin{theorem}[Stuttering] \label{thm:ck-stuttering}
Suppose that $\CSequence := \CSequence_1. \Call. \CSequence_2$ and
$\CSequenced := \CSequence_1. \Call. \Call. \CSequence_2$.
Then for all formulas $\phi \in {\mathcal L}^{ck}$,
$(\Model, \CSequence) \models \phi$ iff $(\Model, \CSequenced) \models \phi$.
\end{theorem}
\begin{proof}
This is a direct consequence of the corresponding 
Stuttering Theorem 3 from \cite{AW16} and Note \ref{not:G}.
\end{proof}

\begin{proof}
  We first prove that a gossip protocol may fail to terminate iff it
  can generate a call sequence $\CSequence.\Call$ such that
  $\CSequence$ is pair-epistemically non-redundant and
  $\EVp{\CSequence.\Call} = \EVp{\CSequence}$. 

\NI 
$(\Ra)$ Let $\overline{\CSequence}$ be an infinite sequence of calls
generated by the protocol. There are only finitely many epistemic
pair-views, so some prefix $\CSequence$ of $\overline{\CSequence}$ is
pair-epistemically redundant.  The claim now follows by Lemma \ref{lem:redundant-common}.
\II

\NI
$(\La)$ Suppose that the protocol generates a sequence of calls
$\CSequence.\Call$ such that $\CSequence$ is pair-epistemically
non-redundant and $\EVp{\CSequence.\Call} = \EVp{\CSequence}$.

Let $\phi$ be the guard associated with the call $\Call$, i.e., 
$\phi \to \Call$ is a rule used in the considered protocol.
By assumption $(\Model, \CSequence) \models \phi$,
so by Lemma \ref{lem:ck-ev}
$(\Model, \CSequence.\Call) \models \phi$. 
By the repeated application of the Stuttering Theorem \ref{thm:ck-stuttering}
we get that for all $i \geq 1$, $(\Model, \CSequence.\Call^{i}) \models \phi$.
Consequently, $\CSequence.\Call^\omega$ is an infinite sequence of
calls that can be generated by the protocol. 
\II

The above equivalence shows that determining whether the protocol
always terminates is equivalent to checking that it cannot generate a
call sequence $\CSequence.\Call$ such that $\CSequence$ is
pair-epistemically non-redundant and $\EVp{\CSequence.\Call} = \EVp{\CSequence}$.

But given a call sequence, by the Decidability of Semantics Theorem
\ref{thm:ck-decidability}, it is decidable to determine whether it can be
generated by the protocol and by Lemma \ref{lem:ab-eff} it is decidable
to determine whether a call sequence is pair-epistemically
non-redundant. Further, by Lemma \ref{lem:ck-finite-non-epi-redundant}
there are only finitely many pair-epistemically non-redundant call
sequences, so the claim follows.
\end{proof}

\end{document}